\documentclass[final]{siamltex}

\usepackage{amsmath,amssymb,euscript,yfonts,psfrag,latexsym,dsfont,graphicx,bbm,color,amstext,wasysym,pdfsync,thumbpdf,pmat}
\usepackage[usenames,dvipsnames]{xcolor}
\usepackage{longtable}
\usepackage{makecell}
\usepackage{slashbox,booktabs}

\definecolor{BrickRed}{rgb}{0.9,0,0}
\definecolor{RoyalBlue}{rgb}{0,0.5,0.7}
\definecolor{Gray}{rgb}{0.3,0.3,0.6}

\newtheorem{thm}{Theorem}
\newtheorem{cor}[thm]{Corollary}
\newtheorem{prop}[thm]{Proposition}
\newtheorem{problem}{Problem}

\newcommand{\Sym}{{\mathbf S}}
\newcommand{\Herm}{{\mathbf H}}

\newcommand{\cS}{{\mathcal S}}

\newcommand{\cR}{{\mathcal R}}
\newcommand{\cL}{{\mathcal L}}

\newcommand{\cN}{{\mathcal N}}

\newcommand{\cE}{{\mathcal E}}

\newcommand{\mR}{{\mathbb R}}
\newcommand{\mC}{{\mathbb C}}

\newcommand{\bx}{{\mathbf x}}
\newcommand{\bhx}{{\hat{\mathbf x} }}
\newcommand{\btx}{{\tilde{\mathbf x}}}

\newcommand{\succe}{\succ_{\hspace*{-1pt}_e}}
\newcommand{\succeqe}{\succeq_{\hspace*{-1pt}_e}}
\newcommand{\prece}{\prec_{\hspace*{-1pt}_e}}
\newcommand{\preceqe}{\preceq_{\hspace*{-1pt}_e}}
\newcommand{\ii}{{\rm i}}

\newcommand{\mr}{{\operatorname{mr}}}

\newcommand{\mrdual}{{\operatorname{mr_{dual}}}}
\newcommand{\trace}{{\operatorname{trace}}}

\title{Linear models based on noisy data\\ and the Frisch scheme}

\author{Lipeng Ning\thanks{L. Ning is with the Dept.\ of Electrical \& Comp.\ Eng.,
University of Minnesota, Minneapolis, Minnesota 55455, {\tt
ningx015@umn.edu}}\and Tryphon T. Georgiou\thanks{T. T.  Georgiou is
with the Dept.\ of Electrical \& Comp.\ Eng., University of
Minnesota, Minneapolis, Minnesota 55455, {\tt tryphon@umn.edu}}\and
Allen Tannenbaum\thanks{A. Tannenbaum is with the Comprehensive
Cancer Center and Dept.\ of Electrical \& Comp.\ Eng., University of
Alabama, Birmingham, AL 35294, {\tt tannenba@uab.edu}} \and {Stephen~P.~Boyd}\thanks{S. P. Boyd is with the Department of Electrical Engineering, Stanford University, Stanford, CA 94305, {\tt boyd@stanford.edu}}
}
\begin{document}

\maketitle

\begin{abstract}
We address the problem of identifying linear relations among
variables based on noisy measurements. This is, of course, a central
question in problems involving ``Big Data.'' Often a key assumption
is that measurement errors in each variable are independent. This
precise formulation has its roots in the work of Charles Spearman in
1904 and of Ragnar Frisch in the 1930's. Various topics such as
errors-in-variables, factor analysis, and instrumental variables,
all refer to alternative formulations of the problem of how to
account for the anticipated way that noise enters in the data. In
the present paper we begin by describing the basic theory and
provide alternative modern proofs to some key results. We then go on
to consider certain generalizations of the theory as well applying
certain novel numerical techniques to the problem.
 A central role is played by the Frisch-Kalman dictum which aims at a noise contribution that allows
a maximal set of simultaneous linear relations among the noise-free
variables --a rank minimization problem. In the years since
Frisch's original formulation, there have been several insights
including trace minimization as a convenient heuristic to replace
rank minimization. We discuss convex relaxations and
certificates guaranteeing global optimality. A complementary point
of view to the Frisch-Kalman dictum is introduced in which models
lead to a min-max quadratic estimation error for the error-free
variables. Points of contact between the two formalisms are
discussed and various alternative regularization schemes are indicated.
\end{abstract}

\section{Introduction}

The standard paradigm in modeling is to postulate that measured quantities contain a contribution of ``accidental deviation''  \cite{Spearman} from the otherwise ``uniformities'' that characterize an underlying law.
Therefore, a key issue when identifying dependencies between variables is how to account for the contribution of noise in the data. Various assumptions on the structure of noise and of the possible dependencies lead to a number of corresponding methodologies.

The purpose of the present paper is to consider from a modern
computational point of view, the important situation where the noise
components are assumed independent, and the consequences of this
assumption --the data is typically abstracted into a corresponding
(estimated) covariance statistic. This independence assumption
underlies the errors-in-variables model \cite{Durbin,KlepperLeamer}
and factor analysis
\cite{AndersonRubin,Ledermann,Harman1966,Joreskog1969,Shapiro}, and
has a century-old history \cite{Frisch2,Reiersol,Koopmans}; see also
\cite{Kalman1982,Kalman1985,Los,Woodgate1,Guidorzi95,Soderstrom2007errors,Anderson2008,Forni2000}.
Accordingly, given the large classical literature on this problem,
this paper will also have a tutorial flavor.

The precise formulation has its roots in the work of Ragnar Frisch
in the 1930's. The central assumption is that the noise components
are independent of the underlying variables and are also mutually
independent \cite{Kalman1982,Kalman1985}. In addition, since several
alternative linear relations are typically consistent with the data,
a maximal set of simultaneous dependencies is sought as a means to
limit uncertainty and to provide canonical models
\cite{Kalman1982,Kalman1985}. This particular dictum gives rise to a
(non-convex) rank-minimization problem. Thus, it is somewhat
surprising that the special case where the maximal number of
possible simultaneous linear relations is equal to $1$ can be
explicitly characterized --this was accomplished over half a
century ago by Reiers{\o}l  \cite{Reiersol}; see also
\cite{Kalman1982,KlepperLeamer}. To date no other case is known that
admits a precise closed-form solution.

In recent years, emphasis has been shifting from hard, non-convex
optimization to convex regularizations, which in addition scale
nicely with the size of the problem. Following this trend we revisit
the Frisch problem from several alternative angles. We first present
an overview of the literature, and present several new insights and
proofs. In the process, we also give an extension of Reiers{\o}l's
result to complex matrices. Our main interest is in exploring
recently studied convex optimization problems that approximate rank
minimization by use of suitable surrogates. In particular, we study
iterative schemes for treating the general Frisch problem and focus
on certificates that guarantee optimality. In parallel, we consider
a viewpoint that serves as an alternative to the Frisch problem
where now, instead of a maximal number of simultaneous linear
relations, we seek a uniformly optimal estimator for the unobserved
data under the independence assumption of the Frisch scheme. The
optimal estimator is obtained as a solution to a min-max
optimization problem. Rank-regularized and min-max alternatives are
discussed and an example is given to highlight the potential and limitations of the techniques.

The remainder of this paper is organized as follows. We first
introduce the errors-in-variables problem in
Section~\ref{sec:datastrcuture}. In Section~\ref{sec:Frisch}, we
revisit the Frisch problem, and a related problem due to Shapiro,
and provide a geometric interpretation of Reiers{\o}l's result along
with a generalization to complex-valued covariances. In
Section~\ref{sec:MinTrace}, we present an iterative
trace-minimization scheme for solving the Frisch problem and provide
computable lower-bounds for the minimum-rank. In
Section~\ref{correspondence}, we bring up the question of estimation
in the context of the Frisch scheme and motivate a suitable a
rank-regularized min-max optimization problem in
Section~\ref{sec:regularized}. Some concluding remarks are provided
in Section~\ref{sec:conclusion}.

\section{Notation}

$\;$\\[.1in]
\noindent
\begin{tabular}{ll}
    $\cR(\cdot)$, $\cN(\cdot)$   & range space, null space\\
    $\Pi_{\mathcal X}$ & orthogonal projection onto ${\mathcal X}$\\
    $>0\;\; (\geq 0)$ & positive definite (resp., positive semi-definite) \\
    $\Sym_n$& $=\;\;\left\{M \mid M\in \mR^{n\times n},\; M=M' \right\}$\\
    $\Sym_{n,+}$& $=\;\;\left\{M \mid M\in \Sym_n,\; M\geq0 \right\}$\\
    $\Herm_n$& $=\;\;\left\{M \mid M\in \mC^{n\times n},\; M=M^* \right\}$\\
    $\Herm_{n,+}$& $=\;\;\left\{M \mid M\in \Herm_n,\; M\geq0 \right\}$\\
    $[\cdot ]_{k\ell},\;\; ([\cdot ]_{k})$ & $(k, \ell)$-th entry (resp., $k$-th entry)\\
    $|M|$& determinant of $M\in \mR^{n\times n}$\\
    $n_+(\cdot)$& number of positive eigenvalues\\
    $\diag: \mR^{n\times n} \to \mR^n: M\mapsto d$ &  where $[d]_i=[M]_{ii}$ for $i=1, \ldots n$\\
    $\diag^*: \mR^{n} \rightarrow \mR^{n\times n}:  d\mapsto D$&  where $D$ is diagonal and $[D]_{ii}=[d]_{i}$ for $i=1,\ldots n$\\
    $M\succe 0\;(\succeqe 0,\; \prece 0,\;\preceqe 0)$& the off-diagonal entries are $>0$ (resp.\ $\geq 0$, $<0$, $\leq 0$),\\&or 
    can be made so by changing the signs of selected\\&rows and corresponding columns\\
\end{tabular}

\section{Data and basic assumptions}\label{sec:datastrcuture}

Consider a Gaussian vector $\bx$ taking values in $\mR^{n\times 1}$ having zero mean and covariance $\Sigma$.  We
assume that it represents an additive mixture of a Gaussian ``noise-free'' vector $\bhx$
and a ``noise component'' $\btx$, thus
\begin{equation}\label{eq:xa}
\bx=\bhx+\btx.
\end{equation}
The entries of $\btx$ are assumed independent of one another
and independent of the entries of $\bhx$ with both vectors having zero mean
and covariances $\hat\Sigma$ and $\tilde\Sigma$, respectively.
Thus,
\begin{subequations}\label{eq:firstsetofconstraints}
\begin{eqnarray}
&&\cE(\btx \btx') =: \tilde\Sigma \mbox{ is diagonal} \label{eq:xc}\\
&&\cE(\bhx \btx')=0. \label{eq:xb}
\end{eqnarray}
Throughout $\cE(\cdot)$ denotes the expectation operation and $0$
denotes the zero vector/matrix of appropriate size. The noise-free
entries of $\bhx$ are assumed to satisfy a set of $q$ simultaneous
linear relations. Hence, $M'\bhx=0$, with $M\in \mR^{n\times q}$ and
$n>\rank(M)=q>0$. The problem is mainly to infer these relations.
Equivalently, $\cE(\bhx \bhx') =: \hat\Sigma$ has
\begin{eqnarray}
&&\rank(\hat\Sigma)= n-q \label{eq:xd}
\end{eqnarray}
\end{subequations}
and $\hat\Sigma M=0$. Statistics are typically estimated from
observation records. To this end, consider a sequence
\[
x_t\in\mR^{n\times 1},\; t=1,\ldots,T
\]
of independent measurements (realizations) of $\bx$
and, likewise, let $\hat x_t$ and $\tilde x_t$ represent the corresponding values of the noise-free
variable and noise components. Denote by
\[
X=\left[\begin{matrix} x_1\;x_2\; \ldots\; x_T\end{matrix}\right]\in \mR^{n\times T}
\]
the matrix of observations of $\bx$ and similarly denote by $\hat X$
and $\tilde X$ the corresponding matrices of the noise-free and
noise entries, respectively. Data for identifying relations among
the noise-free variables are typically limited to the observation
matrix $X$ and, neglecting a scaling factor of $1/T$, the data is
typically abstracted in the form of a sample covariance $XX^\prime$.
For the most part we will assume that sample covariances are
accurate approximations of true covariances, and hence the modeling
assumptions amount to
\begin{subequations}
\begin{eqnarray}
&& \tilde X \tilde X ^\prime  \simeq \mbox{ diagonal}\label{eq:diagonal}\\
&& \hat X  \tilde X ^\prime\simeq  0 \label{eq:orthogonality}\\
&&\rank(\hat X) =n-q \label{eq:rank}
\end{eqnarray}
\end{subequations}
since $M^\prime \hat X=0$.

The number of possible linear relations among the noise free
variables and the corresponding coefficient matrix need to be
determined from either $X$ or $\Sigma$. This motivates the Frisch
and Shapiro problems discussed in Section~\ref{sec:Frisch}. An
alternative set of problems can be motivated by the need to
determine $\hat X$ from $X$ via suitable decomposition
\begin{equation}\label{eq:decompose}
X=\hat X+\tilde X
\end{equation}
in a way that is consistent with the existence of a set of $q$
linear relations. We will return to this in
Section~\ref{sec:min-max}.

\section{The problems of Frisch and Shapiro}\label{sec:Frisch}

We begin with the Frisch problem concerning the decomposition of a
covariance matrix $\Sigma$ that is consistent with the assumptions
in Section~\ref{sec:datastrcuture}. The fact that, in practice,
$\Sigma$ is an empirical sample covariance motivates relaxing
(\ref{eq:xc}-\ref{eq:xd}) in various ways. In particular, relaxation
of the constraint $\tilde \Sigma\geq 0$ leads to the Shapiro
problem.

\begin{problem}[\em The Frisch problem]\label{problem1} Given $\Sigma\in\Sym_{n,+}$, determine
\begin{eqnarray}\nonumber
\mr_+(\Sigma)&:=&\min\{\rank(\hat\Sigma) \mid \Sigma=\tilde \Sigma+\hat \Sigma,\\&&\tilde\Sigma, \hat\Sigma\geq 0,\;\tilde\Sigma \mbox{ is diagonal}\}.\label{eq:mc}
\end{eqnarray}
\end{problem}

\begin{problem}[\em The Shapiro problem]\label{problemShapiro} Given $\Sigma\in\Sym_{n,+}$, determine
\begin{eqnarray}\nonumber
\mr(\Sigma)&:=&\min\{\rank(\hat\Sigma) \mid \Sigma=\tilde \Sigma+\hat \Sigma,\\&& \hat\Sigma\geq 0,\;\tilde\Sigma \mbox{ is diagonal}\}.\label{eq:mc2}
\end{eqnarray}
\end{problem}

The Frisch problem was studied by several researchers, see e.g.,
\cite{Kalman1985,Los,Woodgate1,woodgate2} and the references therein. On the other hand, Shapiro \cite{Shapiro} introduced the above relaxed
version, removing the requirement that $\tilde \Sigma\geq 0$, in an
attempt to gain understanding of the algebraic constraints imposed
by the off-diagonal elements of $\Sigma$ on the decomposition. We
refer to $\mr_+(\cdot)$ as the {\em Frisch minimum rank} and
$\mr(\cdot)$ as the {\em Shapiro minimum rank}. The former is lower
semicontinuous whereas the latter is not, as stated next. This
difference is crucial if one wants to apply this type of methodology
to real data, namely some sort of continuity is necessary.

\begin{prop}\label{lemma:lowersc}
$\mr_+(\cdot)$ is lower semicontinuous whereas $\mr(\cdot)$ is not.
\end{prop}

\begin{proof}
Assume that for a given $\Sigma>0$ there exists a sequence $\Sigma_1,\,\Sigma_2,\,\ldots$ of positive definite matrices such that
$\Sigma_i\rightarrow \Sigma$
while
\[
\mr_+(\Sigma_i)<\mr_+(\Sigma)=r,\; \mbox{ for all }i=1,\,2,\,\dots.
\]
Decompose $\Sigma_i=\hat\Sigma_i+D_i$ with $\rank(\hat\Sigma_i)<r$, $\Sigma_i\geq D_i\geq 0$ and $D_i$ diagonal.
Then there exist convergent subsequences $\hat\Sigma_{i_k}\rightarrow \hat\Sigma$ and $D_{i_k}\rightarrow D$, as $k\to \infty$. Since
 $\Sigma_{i_k}\rightarrow \hat\Sigma+D=\Sigma$, by the lower semicontinuity of the rank,
\[
\rank(\hat\Sigma)\leq \lim_{k\rightarrow \infty}\inf \rank(\hat\Sigma_{i_k})<r=\mr_+(\Sigma).
\]
This is a contradiction.
On the other hand, to see that $\mr(\cdot)$ is not lower semicontinuous consider
\[
\Sigma= \left[\begin{matrix} 3 & -1 &-1 \\ -1 &3 & 0\\ -1 & 0 &3\end{matrix}\right] \mbox{ and }\Sigma_\epsilon
=\left[\begin{matrix} 3 & -1 &-1 \\ -1 &3 & \epsilon \\ -1 & \epsilon &3\end{matrix}\right],\;\
{\hat\Sigma_\epsilon}=\left[\begin{matrix} \frac{1}{\epsilon} & -1 &-1 \\ -1 &\epsilon & \epsilon\\ -1 & \epsilon &\epsilon\end{matrix}\right]
\]
for $\epsilon>0$.
Clearly $\mr(\Sigma)=2$. Also $\lim_{\epsilon\to 0}\Sigma_\epsilon=\Sigma$. Yet $\Sigma_\epsilon=
\hat\Sigma_\epsilon+D_\epsilon$ while $\Sigma_\epsilon$
has rank $1$ and $D_\epsilon$ is diagonal ($\not\geq 0$). Hence $\mr(\Sigma_\epsilon)=1$.
\end{proof}

Assuming that the off-diagonal entries of $\Sigma>0$ of size
$n\times n$ are known with absolute certainty, any ``minimum rank''
($\mr_+(\cdot)$ and $\mr(\cdot)$) is bounded below by the so-called
Lederman bound, i.e.,
\begin{align}\label{ledermann}
\frac{2n+1-\sqrt{8n+1}}{2}\leq \mr(\Sigma)\leq \mr_+(\Sigma),
\end{align}
which holds on a generic set of positive definite matrices $\Sigma$,
that is, on a (Zariski open) subset of positive definite matrices.
Equivalently, the set of matrices $\Sigma$ for which $\mr(\Sigma)$
is lower than the Lederman bound is non-generic --their entries
satisfy algebraic equations which fail under small perturbation. To
see this, consider any factorization
\[\Sigma =FF^\prime,
\]
with $F\in\mR^{n\times r}$. There are $(n-r)r + \frac{r(r+1)}{2}$ independent entries in $F$ (when accounting for the action of a unitary transformation of $F$ on the right), whereas the value of the off-diagonal entries of $\Sigma$ impose $\frac{n(n-1)}{2}$ constraints. Thus, the number of independent entries in $F$ exceeds the number of constraints when $(n-r)^2\geq n+r$ which then leads to the inequality $\frac{2n+1-\sqrt{8n+1}}{2}\leq r$. The bound was first noted in \cite{Ledermann} while the independence of the constraints has been detailed in \cite{Bekker1997}.
In general, the computation of the exact value for $\mr_+(\Sigma)$ and $\mr(\Sigma)$ is a non-trivial matter.
Thus, it is rather surprising that an exact analytic result is available for both, in the special case when $r=n-1$.
We review this next in the form of two theorems.

\begin{thm}[\em Reiers\o l's theorem \cite{Reiersol}] \label{thm:Reiersol}
Let $\Sigma\in \Sym_{n,+}$ and $\Sigma>0$, then
\[
\mr_+(\Sigma)=n-1 \Leftrightarrow \Sigma^{-1} \succe 0.
\]
\end{thm}

\begin{thm}[\em Shapiro's theorem \cite{Shapiro1982b}]\label{thm:Shapiro}
Let $\Sigma\in\Sym_{n,+}$ and irreducible,
\[
\mr(\Sigma)=n-1\Leftrightarrow \Sigma\preceqe 0.
\]
\end{thm}

The characterization of covariance matrices $\Sigma$ for which
$\mr_+(\Sigma)=n-1$ was first recognized by T.~C.~Koopmans in 1937
\cite{Koopmans} and proven by Reiers\o l \cite{Reiersol} who used
the Perron-Frobenius theory to improve on Koopmans' analysis. Later
on, R.~E.~Kalman streamlined and completed the steps in
\cite{Kalman1982} relying again on the Perron-Frobenius theorem (see
also Klepper and Leamer \cite{KlepperLeamer} for a detailed
analysis). Our treatment below takes a slightly different angle and
provides some geometric insight by pointing as a key reason that the
maximal number of vectors at an obtuse angle from one another can
exceed the dimension of the ambient space by at most one
(Corollary~\ref{cor:numberofobtuseangles}). We provide new proofs
where we also utilize a dual formulation with an analogous
decomposition of the inverse covariance.

\subsection{A geometric insight}

We begin with two basic lemmas for irreducible matrices in $M\in\Sym_{n,+}$. Recall that a matrix is reducible if by permutation of rows and columns can be brought into a block diagonal form, otherwise it is irreducible.
\begin{lemma}\label{lemma:previous} Let $M>0$ and irreducible. Then,
\begin{eqnarray}\label{eq:first}
M\preceqe 0 &\Rightarrow & M^{-1}\succe 0.
\end{eqnarray}
\end{lemma}
\begin{lemma}\label{lemma:next} Let $M\geq 0$ and irreducible. Then,
\begin{eqnarray}\label{eq:nullitybound}
M\preceqe 0
&\Rightarrow & {\rm nullity}(M)\leq 1.
\end{eqnarray}
\end{lemma}
\begin{proof} 
It is easy to verify that for matrices of size $2\times 2$,
(\ref{eq:first}) holds true. Assume that the statement also holds true for matrices of size up to $k\times k$, for a certain value of $k$,
and consider a matrix $M$ of size $(k+1)\times(k+1)$ with $M>0$ and $M\preceqe 0$. Partition
\[M=\left[\begin{matrix}A &b\\b' &c\end{matrix}\right]
\]
so that $c$ is a scalar and, hence, $A$ is of size $k\times k$.
Partitioning conformably,
\[M^{-1}=\left[\begin{matrix}F &g\\g' &h\end{matrix}\right]
\]
where
\[F=(A-bc^{-1}b')^{-1}, ~g=-A^{-1}bh, \mbox{ and }h=(c-b'A^{-1}b)^{-1}>0.
\]

For the case where $A$ is irreducible, because $A$ has size $k\times
k$ and $A\preceqe 0$,  invoking our hypothesis we conclude that
$A^{-1}\succe  0$. Now, since $b$ has only non-positive entries and
$b\neq0$, $g=-A^{-1}bh$ has positive entries. Since
$-bc^{-1}b'\preceqe 0$ and $A\preceqe 0$, then $A-bc^{-1}b'\preceqe
0$ is also irreducible. Thus $F=(A-bc^{-1}b')^{-1}$ has positive
entries by hypothesis.

For the case where $A$ is reducible, permutation of columns and rows
brings $A$ into a block-diagonal form with irreducible blocks. Thus,
$A^{-1}$ is also block diagonal matrix with each block entry-wise
positive. Because $M$ is irreducible,  $b$ must have at least one
non-zero entry corresponding to the rows of each diagonal blocks of
$A$. Then $A-bc^{-1}b'$ is irreducible and $\preceqe 0$. Also
$A^{-1}b$ has all of its entries negative. Therefore
$F=(A-bc^{-1}b')^{-1}$ and $g=-A^{-1}bh$ have positive entries.
Therefore $M^{-1}\succe 0$.
\end{proof}

\begin{proof} 
Rearrange rows and columns and partition
\[M=\left[\begin{matrix}A &B\\B' &C\end{matrix}\right]
\]
so that $A$ is nonsingular and of maximal size, equal to the rank of $M$.
Then
\begin{equation}\label{eq:equality}
C=B'A^{-1}B.
\end{equation}

We first show that $B'A^{-1}B\succeqe 0$. Assume that $A$ is irreducible.
Then $A^{-1}\succe 0$. At the same time $B$ has negative entries and not all zero (since $M$ is irreducible). In this case, $B'A^{-1}B\succe 0$.
If on the other hand $A$ is reducible,  Lemma \ref{lemma:previous} applied to the (irreducible) blocks of $A$ implies that $A^{-1}\succeqe 0$.
Therefore, in this case, $B'A^{-1}B\succeqe 0$.

Returning to \eqref{eq:equality} and in view of the fact that $C\preceqe 0$ while $B'A^{-1}B\succeqe 0$ we conclude that, either $C$ is a scalar (and hence there are no off-diagonal negative entries), or both $C$ and $B'A^{-1}B$ are diagonal. The latter contradicts the assumption that $M$ is irreducible. Hence, the nullity of $M$ can be at most $1$.
\end{proof}

Lemma \ref{lemma:next} provides the following geometric insight, stated as a corollary.
\begin{cor}\label{cor:numberofobtuseangles} In any Euclidean space of dimension $n$,
there can be at most $n+1$ vectors forming an obtuse angle with one another.
\end{cor}

\begin{proof} The Grammian $M=[v_k'v_\ell]_{k,\ell=1}^{n+q}$ of a selection $\{v_k\mid k=1,\ldots, n+q\}$ of such vectors
has off-diagonal entries which are negative. Hence, by Lemma \ref{lemma:next}, the nullity of $M$ cannot exceed $1$.
\end{proof}

The necessity part of Theorem \ref{thm:Shapiro} is also a direct corollary of Lemma \ref{lemma:next}.
\begin{cor}\label{prop:weaker} Let $\Sigma\in\Sym_{n,+}$ and irreducible. Then
\[
\Sigma \preceqe 0 \Rightarrow\mr(\Sigma)=n-1.
\]
\end{cor}

\begin{proof}
Let
$\Sigma =\hat \Sigma+\tilde\Sigma$, with $\tilde\Sigma$ diagonal and $\hat\Sigma\geq 0$. $\hat\Sigma$ is irreducible since $\Sigma$ is irreducible. From Lemma \ref{lemma:next}, the nullity of $\hat\Sigma$ is at most $1$. Thus $\mr(\Sigma)=n-1$.
\end{proof}

\subsection{A dual decomposition}

The matrix inversion lemma provides a correspondence between an
additive decomposition of a positive-definite matrix and a
decomposition of its inverse, albeit with a different sign in one of
the summands. This is stated next.

\begin{lemma}\label{lemma:decompositions} Let
\begin{equation}\label{eq:first_decomposition}
\Sigma=D+FF'
\end{equation}
with $\Sigma,D\in\Sym_{n,+}$,
with $\Sigma,D>0$ and $F\in\mR^{n\times r}$. Then
\begin{equation}\label{eq:second_decomposition}
S:=\Sigma^{-1} = E - GG'
\end{equation}
for $E=D^{-1}$ and $G=D^{-1}F(I+F'D^{-1}F)^{-1/2}$. Conversely, if (\ref{eq:second_decomposition}) holds with $G\in\mR^{n\times r}$, then so does (\ref{eq:first_decomposition}) for $D=E^{-1}$ and $F=E^{-1}G(I-G'E^{-1}G)^{-1/2}$.
\end{lemma}

\begin{proof} This follows from the identity
$(I\pm MM')^{-1}=I\mp M(I\mp M'M)^{-1}M'$.
\end{proof}

Application of the lemma suggests the following variation to Frisch's problem.
\begin{problem}[\em The dual Frisch problem]\label{problem2} Given a positive-definite $n\times n$ symmetric matrix $S$ determine
the {\em dual minimum rank}:
\begin{eqnarray}\nonumber
\mrdual(S)&:=&\min\{\rank(\hat S \mid S=E -\hat S,\\&& \hat S,E\geq 0,\;E \mbox{ is diagonal}\nonumber\}.\label{eq:mcdual}
\end{eqnarray}
\end{problem}

Clearly, if $S=\Sigma^{-1}=E-GG^\prime$ (as in (\ref{eq:second_decomposition})), then $E>0$. Furthermore, a
decomposition of $S$ always gives rise to
a decomposition $\Sigma=D+FF^\prime$ (as in (\ref{eq:first_decomposition})) with the terms $FF'$ and
$GG'$ having the same rank. Thus, it is clear that
\begin{equation}\label{eq:inequality}
\mr_+(\Sigma)\leq \mrdual(\Sigma^{-1}),
\end{equation}
and that the above holds with equality when an optimal choice of $D\equiv\tilde \Sigma$ in (\ref{eq:mc}) is invertible.
However, if $D$ is allowed to be singular,
the rank of the summands $FF'$ and $GG'$ may not agree. This is can be seen using the following example. Take
\[\Sigma=\left[\begin{matrix}
2&1&1\\
1&2&1\\
1&1&1\end{matrix}\right].
\]
It is clear that $\Sigma$ admits a decomposition
$\Sigma=\tilde\Sigma+\hat\Sigma$, in correspondence with
(\ref{eq:first_decomposition}), where
$\tilde\Sigma=D=\diag\{1,1,0\}$
while $\hat\Sigma=FF'$ as well as $F'=[1,\,1,\,1]$ are of rank
one. On the other hand,
\[
S=\Sigma^{-1}=\left[\begin{matrix}
\;\;1&\;\;0&-1\\
\;\;0&\;\;1&-1\\
-1&-1&\;\;3\end{matrix}\right].
\]
Taking $E=\diag\{e_1,\;e_2,\;e_3\}$ in (\ref{eq:second_decomposition}), it is evident that the rank of
\[
GG'=E-S=\left[\begin{matrix}
e_1-1&0&1\\
0&e_2-1&1\\
1&1&e_3-3\end{matrix}\right]
\]
cannot be less than $2$ without violating the non-negativity
assumption for the summand $GG'$. The minimal rank for the factor
$G$ is $2$ and is attained by taking $e_1=e_2=2$ and $e_3=5$.

On the other hand, in general, if we perturb $\Sigma$ to $\Sigma+\epsilon I$ and, accordingly, $D$ to $D+\epsilon I$, then
\begin{equation}\label{eq:inequality2}
\mrdual((\Sigma+\epsilon I)^{-1})\leq \mr_+(\Sigma),~ \forall \epsilon>0.
\end{equation}
Equality in \eqref{eq:inequality2} holds for sufficiently small
value of $\epsilon$. Thus, $\mr_+$ and $\mrdual$ are closely
related. However, it should be noted that $\mrdual(\cdot)$ fails to
be lower semi-continuous since a small perturbation of the
off-diagonal entries can reduce $\mrdual(\cdot)$. Yet,
interestingly, an exact characterization of the $\mrdual(S)=n-1$ can
be obtained which is analogous to those for $\mr_+$ and $\mr$ being
equal to $n-1$; the condition for $\mrdual$ will be used to prove
the Reiers\o l and Shapiro theorems.

\begin{thm}\label{thm:dualreiersol} For $S\in\Sym_{n,+}$, with $S>0$ and irreducible,
\begin{equation}\label{dualreiersol}
\mrdual(S)=n-1 \Leftrightarrow S \succeqe 0.
\end{equation}
\end{thm}

\begin{proof} 
If $S\succeqe 0$ and $E$ is diagonal satisfying $E\geq S>0$, then $E-S=GG'\preceqe 0$.
By invoking Lemma~\ref{lemma:next} we deduce that if $E-S$ is singular, $\rank(G)=n-1$. Hence, $\mrdual(S)=n-1$.

To establish that $\mrdual(S)=n-1\Rightarrow S \succeqe 0$, we
assume that the condition $S\succeqe 0$ fails and show that
$\mrdual(S)<n-1$. We first argue the case for a $3\times 3$ matrix
$S=[s_{ij}]_{i,j=1}^3$. Provided $S\not\succeqe 0$ we can 
assume that it has strictly negative off-diagonal entries (which can be done by
reflecting the signs of rows and columns).
We now let
\begin{eqnarray*}
e_{i}&=&s_{ii}-\frac{s_{ij}s_{ki}}{s_{jk}}
\end{eqnarray*}
for $i\in\{1,2,3\}$ and $(i,j,k)$ being permutations of $(1,2,3)$.
These are all positive. Let $\tilde S=\diag^*(e_1,e_2,e_3)$.
It can be seen that $\tilde S-S\geq0$ while $\rank(\tilde S-S)=1$.
To verify the latter observe that $\tilde S-S=vv'$
for
\[v'=\left[\begin{matrix}
\sqrt{e_1-s_{11}},&
\sqrt{e_2-s_{22}},&
\sqrt{e_3-s_{33}}\end{matrix}\right].
\]
This establishes the reverse implication for matrices of size $3\times 3$.

We now assume that the statement holds true for matrices of size up to $(n-1)\times (n-1)$ for some $n\geq 4$ and use induction.
So let $S,\;\tilde S$ be of size $n\times n$ with $S\not\succeqe 0$ and $\tilde S$ diagonal. We need to prove that $\mrdual(S)<n-1$.
We partition
\[
S=\left[\begin{matrix} A&b\\ b'&c \end{matrix}
\right]
,\; \tilde S=\left[\begin{matrix} E&0\\ 0& e\end{matrix}\right]
\]
with $A,\;E$ being $(n-1)\times (n-1)$.
For any $\tilde S$ such that $\tilde S-S\geq 0$, $e$ cannot be equal to $c$, otherwise $b=0$ and $S$ is reducible.
Further, $\tilde S-S\geq 0$ if and only if $e>c$ and
\[
M:=E-(A+b(e-c)^{-1}b')\geq0.
\]
The nullity of $\tilde S-S$ coincides with that of $M$. To prove our claim, it suffices to show
that $A_e:=A+b(e-c)^{-1}b'\not\succeqe 0$, or that $A_e$ is reducible for some $e>c$. (Since, in either case, by our hypothesis, the nullity of $M$ for a suitable $E$ exceeds $1$.)

We now consider two possible cases where $S\succeqe 0$ fails. First, we consider the case where already $A\not \succeqe 0$.
Then obviously $A_e\not\succeqe 0$ for $e-c$ sufficiently large.
The second possibility is $S\not \succeqe 0$ while $A\succeqe 0$. But if $A$ is (transformed into) element-wise nonnegative, then $bb'$ must have at least one pair of negative off-diagonal entries. Then, consider $A_e=A+\lambda bb'$ for $\lambda=(e-c)^{-1}\in(0,\infty)$. Evidently, for certain values of $\lambda$ entries of $A_e$ change sign. If a whole row becomes zero for a particular value of $\lambda$, then $A_e$ is reducible. In all other cases, there are values of $\lambda$ for which $A_e\not\succeqe 0$. This completes the proof.
\end{proof}

\subsection{Proof of Reiers{\o}l's theorem (Theorem \ref{thm:Reiersol})}\label{sec:proofReiersol}

We first show that $\Sigma^{-1}\succe 0$ implies
$\mr_+(\Sigma)=n-1$. From the continuity of the inverse,
$(\Sigma+\epsilon  I)^{-1}\succe 0$ for sufficiently small
$\epsilon>0$. Applying Theorem~\ref{thm:dualreiersol}, we conclude
that
\[
\mrdual((\Sigma+\epsilon  I)^{-1})=n-1.
\]
Since  $\mr_+(\Sigma)\geq \mrdual((\Sigma+\epsilon  I)^{-1})$ as in
\eqref{eq:inequality2}, we conclude that $\mr_+(\Sigma)=n-1$.

To prove that $\mr_+(\Sigma) =n-1\Rightarrow \Sigma^{-1}\succe 0$,
we show that assuming $\Sigma^{-1}\not\succe 0$ and $\mr_+(\Sigma)
=n-1$ together leads to a contradiction. From the continuity of the
inverse and the lower semicontinuity of $\mr_+(\cdot)$ (Proposition
\ref{lemma:lowersc}), there exists a symmetric matrix $\Delta$ and
an $\epsilon>0$ such that
\[
(\Sigma+\epsilon  \Delta)^{-1} \not \succeqe 0, \text{~and~} \mr_+(\Sigma+\epsilon  \Delta)=n-1.
\]
Then, from Theorem \ref{thm:dualreiersol},
$
\mrdual((\Sigma+\epsilon  \Delta)^{-1})< n-1
$
while from \eqref{eq:inequality}
\[
\mr_+(\Sigma+\epsilon  \Delta) \leq \mrdual((\Sigma+\epsilon  \Delta)^{-1}).
\]
Thus, we have a contradiction and therefore $\Sigma^{-1}\succe 0$. $\Box$

\subsection{Proof of Shapiro's theorem (Theorem \ref{thm:Shapiro})}\label{sec:proofShapiro}
Given $\Sigma\geq 0$ consider $\lambda>0$ such that $\lambda I-\Sigma\geq0$, a diagonal $D$, and let $E:=\lambda I-D$.
Since
$\Sigma-D=E-(\lambda I -\Sigma)$,
\begin{align}\label{eq:mrmrdual}
\mr(\Sigma)=\mrdual(\lambda I-\Sigma).
\end{align}
If $\Sigma$ is irreducible and $\Sigma\preceqe 0$, then $\lambda  I-\Sigma$ is irreducible and $\lambda  I-\Sigma\succeqe 0$. It follows (Theorem \ref{thm:dualreiersol}) that
$\mrdual(\lambda I-\Sigma)=n-1$, and therefore $\mr(\Sigma)=n-1$ as well.

For the the reverse direction, if $\mr(\Sigma)=n-1$ then $\mrdual(\lambda I-\Sigma)=n-1$, which implies that
$\lambda  I-\Sigma\succeqe 0$ and therefore that $\Sigma\preceqe 0$. $\Box$

The original proof in \cite{Shapiro1982b} claims that for any $\Sigma\geq 0$ of size $n\times n$ with $n>3$ and $\Sigma\not \preceqe 0$, there exists a $(n-1)\times (n-1)$ principle minor that is $\not\preceqe 0$. This statement fails for the following sign pattern
\[\footnotesize{
\left[\begin{matrix}+&0&-&-\\0&+&-&+\\-&-&+&0\\-&+&0&+    \end{matrix} \right].}
\]
This matrix can not transformed to have all nonpositive off-diagonal entries, yet all its $3\times 3$ principle minors $\preceqe 0$.

\subsection{Parametrization of solutions under Reiers{\o}l's and Shapiro's conditions}\label{section:parametrization}

For either the Frisch or the Shapiro problem, a solution is not
unique in general. The parametrization of solutions to the Frisch
problem when $\mr_+(\Sigma)=n-1$ has been known and is briefly
explained below (without proof). Interestingly, an analogous
parametrization is possible for Shapiro's problem and this is given
in Proposition~\ref{shapiro_parametrization} that follows, and both
are presented here for completeness of the exposition.

\begin{prop}\label{lemma:parameter1}
Let $\Sigma\in\Sym_{n,+}$ with $\Sigma>0$ and $\Sigma^{-1}\succe0$. The following hold:
\begin{itemize}
\item[i)] For $D\geq0$ diagonal with $\Sigma-D\geq0$ and singular, there is a probability vector $\rho$ ($\rho$ has entries $\geq 0$ that sum up to $1$) such that $(\Sigma-D)\Sigma^{-1}\rho=0$.
\item[ii)] For any probability vector $\rho$,
\[D=\diag^*\left(\left[\frac{[\rho]_i}{[\Sigma^{-1}\rho]_i}, i=1,\ldots, n\right] \right)
\]
satisfies $\Sigma-D\geq0$ and $\Sigma-D$ is singular.
\end{itemize}
\end{prop}

\begin{proof} See \cite{Kalman1982,KlepperLeamer}.
\end{proof}

Thus, solutions of Frisch's problem under Reiers{\o}l's conditions
are in bijective correspondence with probability vectors. A very
similar result holds true for Shapiro's problem.

\begin{prop}\label{shapiro_parametrization}
Let $\Sigma\in\Sym_{n,+}$ be irreducible and have $\leq 0$ off-diagonal entries. The following hold:
\begin{itemize}
\item[i)]
For $D$ diagonal with $\Sigma-D\geq0$ and singular, there is a strictly positive vector $v$ such that $(\Sigma-D)v=0$.
\item[ii)] For any strictly positive vector $v\in \mR^{n\times 1}$,
\begin{align}\label{eq:Dshapiro}
D=\diag^*\left(\left[\frac{[\Sigma v]_i}{[v]_i}, i=1,\ldots, n\right] \right)
\end{align}
satisfies that $\Sigma-D\geq0$ and $\Sigma-D$ is singular.
\end{itemize}
\end{prop}

\begin{proof}
To prove $(i)$, we note that if $(\Sigma-D)v=0$, then $v\succe 0$.
To see this consider $(\Sigma-D+\epsilon I)^{-1}$ for $\epsilon>0$.
From Lemma~\ref{lemma:previous},
\[
(\Sigma-D+\epsilon I)^{-1}\succe 0
\]
and since $v$ is an eigenvector corresponding to its largest eigenvalue, a power iteration argument concludes that $v\succe 0$.

To prove $ii)$, it is easy to verify that the diagonal matrix $D$ in \eqref{eq:Dshapiro} for $v\succe 0$
satisfies $(\Sigma-D)v=0$. We only need to prove that $\Sigma-D\geq0$. Without loss of generality we assume that all the entries of $v$ are equal.
(This can always be done by scaling the entries of $v$ and scaling accordingly rows and columns of $\Sigma$.)
Since $v$ is a null vector of $\Sigma-D$ and since $M:=\Sigma-D$ has $\leq 0$ off-diagonal entries
\[
[M]_{ii}=\sum_{j\neq i}|[M]_{ij}|.
\]
Gersgorin Circle Theorem (e.g., see \cite{Varga2004})
now states that every eigenvalue of $M$ lies within at least one of the closed discs $\left\{{\rm Disk}\left([M]_{ii}, \sum_{j\neq i}|[M]_{ij}| \right), i=1, \ldots, n\right\}$. No disc intersects the negative real line. Therefore $\Sigma-D\geq0$.
\end{proof}

\subsection{Decomposition of complex-valued matrices}

Complex-valued covariance matrices are commonly used in radar and
antenna arrays \cite{vantrees}. The rank of $\Sigma-D$, for
noise covariance $D$ as in the Frisch problem, is an indication of
the number of (dominant) scatterers in the scattering field. If this
is of the same order as the number of array elements (e.g., $n-1$),
any conclusion about their location may be suspect. Thus, it is
natural to seek conditions for $\mr_+(\Sigma)=n-1$ analogous to
those given by Reiers{\o}l, for the case of complex covariances, as
a possible warning. This we do next.

Consider complex-valued observation vectors
$
x_t=y_t+ \ii z_t,~ t=1,\ldots T,
$
where $\ii=\sqrt{-1}$ and $y_t, z_t \in \mR^{n\times 1}$, and
set
\[
X=[x_1,\; \ldots x_T]=Y+ \ii Z
\]
with
$Y=[y_1,\; \ldots y_T]$,
$Z=[z_1,\; \ldots z_T]$.
The (scaled) sample covariance is
\begin{align*}
\Sigma=XX^*
&=\Sigma_{\rm r}+\ii \Sigma_{\rm i}\in\Herm_{n,+},
\end{align*}
where the real part
$\Sigma_{\rm r}:=YY'+ZZ'$ is symmetric,
the imaginary part $\Sigma_{\rm i}:=ZY'-YZ'$ is anti-symmetric,
and ``$*$'' denotes complex-conjugate transpose.
As before, we consider a decomposition
\[
\Sigma=\hat\Sigma+D
\]
with $\hat\Sigma\geq 0$ singular and $D\geq 0$ diagonal.
We refer to  \cite{Anderson1988,Deistler1989} for the special case where
$\mr_+(\Sigma)=1$. In this section we present a sufficient condition for a Reiers{\o}l-case
where $\mr_+(\Sigma)=n-1$.

Before we proceed we note that re-casting the problem in terms
of the real-valued
\[
R:=\left[
     \begin{array}{cc}
       \Sigma_{\rm r} & \Sigma_{\rm i} \\
       \Sigma_{\rm i}^\prime & \Sigma_{\rm r} \\
     \end{array}
   \right]\in\Sym_{2n,+}
\]
does not allow taking advantage of earlier results. The structure of $R$ with antisymmetric off-diagonal blocks implies that if $[a',\;b']'$ is a null vector then so is
$[-b',\;a']'$ (since, accordingly, $a+ \ii b$ and $\ii a - b$ are both null vectors of $\Sigma$). Thus, in general, the nullity of $R$ is not $1$ and the theorem of Reiers{\o}l is not applicable. Further, the corresponding noise covariance is diagonal with repeated blocks.

The following lemmas for the complex case echo Lemma \ref{lemma:previous} and Lemma \ref{lemma:next}.
\begin{lemma}\label{lemma:complexprevious}
Let $M\in\Herm_{n,+}$ be irreducible. If the argument of each non-zero off-diagonal entry of $-M$ is in $\left(-\frac{\pi}{2^n},~ \frac{\pi}{2^n} \right)$, then
each entry of $M^{-1}$ has argument in $\left(-\frac{\pi}{2}+\frac{\pi}{2^n}, ~ \frac{\pi}{2}-\frac{\pi}{2^n}\right)$.
\end{lemma}

\begin{proof} 
It is easy to verify the lemma for $2\times 2$ matrices. Assume that
the statement holds for sizes up to $n\times n$ and consider an
$(n+1)\times (n+1)$ matrix $M$ that satisfies the conditions of the
lemma. Partition
\[
M=\left[
    \begin{array}{cc}
      A & b \\
      b^* & c \\
    \end{array}
  \right]
\]
with $A$ is of size $n\times n$, and conformably,
\[
M^{-1}=\left[
    \begin{array}{cc}
      F & g \\
      g^* & h \\
    \end{array}
  \right].
\]
By assumption non-zero entries of $-A$ and $-b$ have their argument in $\left(-\frac{\pi}{2^{n+1}}, ~\frac{\pi}{2^{n+1}}\right)$.
Then, by bounding the possible contribution of the respective terms, it follows that for the argument of each of the entries of $-A+bc^{-1}b^*$ is in
$\left(-\frac{\pi}{2^n}, ~\frac{\pi}{2^n}\right)$. Then, the argument of each entry of $F=(A-bc^{-1}b^*)^{-1}$ is in
$\left(-\frac{\pi}{2}+\frac{\pi}{2^n}, ~ \frac{\pi}{2}-\frac{\pi}{2^n}\right)$; this follows by assumption since $F$ is $n\times n$.
Clearly, $\left(-\frac{\pi}{2}+\frac{\pi}{2^n}, ~ \frac{\pi}{2}-\frac{\pi}{2^n}\right) \subset \left(-\frac{\pi}{2}+\frac{\pi}{2^{n+1}}, ~\frac{\pi}{2}-\frac{\pi}{2^{n+1}}\right)$. Regarding $g$, by bounding the possible contribution of respective terms, we similarly conclude that
the argument of each of its non-zero entries is in $\left(-\frac{\pi}{2}+\frac{\pi}{2^{n+1}}, ~ \frac{\pi}{2}-\frac{\pi}{2^{n+1}}\right)$.
\end{proof}

\begin{lemma}\label{lemma:complexnext}
Let $M\in\Herm_{n,+}$ be irreducible. If the argument of each non-zero off-diagonal entry of $-M$ is in $\left(-\frac{\pi}{2^n},~\frac{\pi}{2^n}\right)$,
then $\rank(M)\geq n-1$.
\end{lemma}

\begin{proof}
First rearrange rows and columns of $M$, and partition as
\[
M=\left[
  \begin{array}{cc}
    A & B \\
    B^* & C \\
  \end{array}
\right]
\]
so that $A$ is nonsingular and of size equal to the rank of $M$, which we denote by $r$. Then
\begin{equation}\label{eq:CBAB}
C=B^*A^{-1}B
\end{equation}
and has size equal to the nullity of $M$. We now compare the
argument of the off-diagonal entries of $C$ and $B^*A^{-1}B$, and
show they cannot be equal unless $C$ is a scalar. Since the
off-diagonal entries of $-A$ have their argument in
$\left(-\frac{\pi}{2^n}, ~\frac{\pi}{2^n}\right)\subset
\left(-\frac{\pi}{2^r}, ~\frac{\pi}{2^r}\right)$, the off-diagonal
entries of $A^{-1}$ have their argument in
$\left(-\frac{\pi}{2}+\frac{\pi}{2^r}, ~
\frac{\pi}{2}-\frac{\pi}{2^r}\right)$ from Lemma
\ref{lemma:complexprevious}. Now, the $(k,\ell)$ entry of
$B^*A^{-1}B$ is
\begin{align*}
[B^*A^{-1}B]_{k\ell}=\sum_{i,j}[B^*]_{ki}[A^{-1}]_{ij}[B]_{j \ell}
\end{align*}
and the phase of each summand is
\[
\arg([B^*]_{ki}[A^{-1}]_{ij} [B]_{j \ell}) \in\left(-\frac{\pi}{2}+\frac{\pi}{2^r}-\frac{\pi}{2^{n-1}},~ \frac{\pi}{2}-\frac{\pi}{2^r}+\frac{\pi}{2^{n-1}}\right).
\]
Thus, the non-zero off-diagonal entries of $B^*A^{-1}B$ have positive real part while
 \[
 \arg(-[C]_{k\ell})\in \left(-\frac{\pi}{2^n},~\frac{\pi}{2^n}\right) .
 \]
Hence, either the off-diagonal entries of $B^*A^{-1}B$ and $C$ are
zero, in which case these are diagonal matrices and $M$ must be
reducible, or $B^*A^{-1}B$ and $C$ are both scalars. This concludes
the proof.
\end{proof}

\begin{thm}\label{prop:complexReiersol}
Let $\Sigma\in\Herm_{n,+}$ be irreducible.
If the argument of each non-zero off-diagonal entry of $-\Sigma$ is in
$\left(-\frac{\pi}{2^n},~ \frac{\pi}{2^n}\right)$, then
$\mr(\Sigma)=n-1$.
\end{thm}

\begin{proof}
The matrix $\Sigma-D$ is irreducible since $D$ is diagonal.
If $\Sigma-D\geq0$ and singular, and since the argument of each non-zero off-diagonal entry of $-(\Sigma-D)$ is in
$\left(-\frac{\pi}{2^n},~ \frac{\pi}{2^n}\right)$, Lemma \ref{lemma:complexnext} applies and gives that $\rank(\Sigma-D)=n-1$.
\end{proof}

Clearly, since $\mr_+(\Sigma)\geq\mr(\Sigma)$, under the condition of Theorem \ref{prop:complexReiersol}, $\mr_+(\Sigma)=n-1$.
It is also clear that for $S\in\Herm_{n,+}$ irreducible with all non-zero off-diagonal entries having argument in $\left(-\frac{\pi}{2^n},~ \frac{\pi}{2^n}\right)$, we also conclude that
$\mrdual(S)=n-1$.

\section{Trace minimization heuristics}\label{sec:MinTrace}

The rank of a matrix is a non-convex function of its elements and the problem to find the matrix of minimal rank within a given set  is a difficult one, in general.
Therefore, certain heuristics have been developed over the years to obtain approximate solutions.
In particular, in the context of factor analysis, trace minimization has been pursued as a suitable heuristic \cite{Ledermann1940,Shapiro,Shapiro1982b} thereby relaxing the Frisch problem into
\begin{align}\nonumber
&\min_{D: \Sigma\geq D\geq0} \trace(\Sigma-D),
\end{align}
for a diagonal matrix $D$; with a relaxation of $D\geq 0$ corresponding to Shapiro's problem. The theoretical basis for using the trace and, more generally, the nuclear norm for non-symmetric matrices, as a surrogate for the rank was provided by
Fazel {\em etal.} \cite{Fazel2001} who proved that these constitute convex envelops of the rank function on  bounded sets of matrices.

The relation between minimum trace factor analysis and minimum rank factor analysis goes back to Ledermann in \cite{Ledermann1939} (see \cite{Della1982} and \cite{Saunderson2012}). Herein we only refer to two propositions which characterize minimizers for the two problems, Frisch's and Shapiro's, respectively.

\begin{subequations}
\begin{prop}[\cite{Della1982}]\label{prop:mintrace}
Let $\Sigma=\hat\Sigma_1+D_1>0$ for a diagonal $D_1\geq0$. Then,
\begin{align}\label{trmc}
 &(\hat\Sigma_1,D_1)=\arg\min\{ \trace(\hat\Sigma) \mid \Sigma=\hat\Sigma+D>0,\;\hat\Sigma\geq 0,\;\mbox{diagonal }D\geq 0\}\\
&  \Leftrightarrow~ \exists~ \Lambda_1 \geq0 ~:~ \hat\Sigma_1 \Lambda_1=0 \text{~and~} \left\{
                         \begin{array}{ll}
                           [\Lambda_1]_{ii}=1, & \text{~if~} [D_1]_{ii}>0, \\
                           \left[ \Lambda_1\right]_{ii}\geq1, &\text{~if~}[D_1]_{ii}=0.\nonumber
                         \end{array}
                       \right.
\end{align}
\end{prop}

\begin{prop}[\cite{Saunderson2012}]\label{prop:MTFAShapiro}
Let $\Sigma=\hat\Sigma_2+D_2>0$ for a diagonal $D_2$.
Then,\begin{align}\label{trmc2}
 &(\hat\Sigma_2,D_2)=\arg\min\{ \trace(\hat\Sigma) \mid \Sigma=\hat\Sigma+D>0,\;\hat\Sigma\geq 0,\;\mbox{diagonal }D\}\\
& \Leftrightarrow~\exists~ \Lambda_2 \geq0 ~:~ \hat\Sigma_2 \Lambda_2=0 \text{~and~} [\Lambda_2]_{ii}=1~ \forall i.\nonumber
\end{align}
\end{prop}
\end{subequations}

\noindent Evidently, when the solutions to these two problems differ
and $D_1\neq D_2$, then there exists $k\in\left\{1, \ldots, n\right\}$ such that
\[
[D_2]_{kk}<0 \text{~and~} [D_1]_{kk}=0.
\]
Further, the essence of Proposition \ref{prop:MTFAShapiro} is that
a singular $\hat\Sigma$ originates from such a minimization problem if and only if there is a correlation matrix in its null space. The matrices $\Lambda_1$ and $\Lambda_2$ appear as Lagrange multipliers in the respective problems.

\newcommand{\cA}{{\mathcal A}}
{Factor analysis is closely related to {\em low-rank matrix completion} as well as to {\em sparse and low-rank decomposition} problems. Typically, low-rank matrix completion asks for a matrix $X$ which satisfies a linear constraint $\cA(X)=b$ and has low/minimal rank ($\cA(\cdot)$ denotes a linear map $\cA\,:\,{\mR}^{n\times n}\rightarrow {\mR}^p$). Thus, factor analysis corresponds to the special case where $\cA(\cdot)$ maps $X$ onto its off-diagonal entries. In a recent work by Recht {\em etal.}~\cite{Recht2010guaranteed}, the nuclear norm of $X$ was considered as a convex relaxation of $\rank(X)$ for such problems and a sufficient condition for exact recovery was provided. However, this sufficient condition amounts to the requirement that the null space of $\cA(\cdot)$ contains no matrix of low-rank. Therefore, since in factor analysis diagonal matrices are in fact contained in the null space of $\cA(\cdot)$ and include matrices of low-rank, the condition in \cite{Recht2010guaranteed} does not apply directly. Other works on low-rank matrix completion (see, e.g., \cite{Recht2010guaranteed,Candes2009exact}) mainly focus on assessing the probability of exact recovery and on constructing efficient computational algorithms for {\em large-scale} low-rank completion problems \cite{Keshavan2010matrix,Keshavan2010noisy}.
On the other hand, since diagonal matrices are sparse (most of their entries are zero), the work on matrix decomposition into sparse and low-rank components by Chandrasekaran {\em etal.} \cite{Chandrasekaran2011rank} is very pertinent. In this, the $\ell_1$ and nuclear norms were used as surrogates for sparsity and rank, respectively, and a sufficient condition for exact recovery was provided which captures a certain ``rank-sparsity incoherence''; an analogous but stronger sufficient ``incoherence'' condition which applies to problem
\eqref{trmc2} is given in \cite{Saunderson2012}.}

\subsection{Weighted minimum trace factor analysis}

Both $\mr(\Sigma)$ and $\mr_+(\Sigma)$ in \eqref{eq:mc} and
\eqref{eq:mc2}, respectively, remain invariant under scaling of rows
and the corresponding columns of $\Sigma$ by the same coefficients.
On the other hand, the minimizers in \eqref{trmc} and \eqref{trmc2}
and their respective ranks are not invariant under scaling. This
fact motivates weighted-trace minimization,
\begin{align}\label{eq:Dw}
\min\left\{ \trace(W\hat\Sigma) \mid \Sigma=\hat\Sigma+D,~\hat\Sigma\geq 0,~\mbox{diagonal }D\geq 0 \right\},
\end{align}
given $\Sigma>0$ and a diagonal weight $W>0$.
As before the characterization of minimizers relates to a suitable condition for the corresponding Lagrange multipliers:

\begin{prop}[{\rm \cite{Shapiro1982b}}]\label{prop:WMTFAShapiro}
Let $\Sigma=\hat\Sigma_0+D_0>0$ for a diagonal matrix $D_0\geq0$ and consider a diagonal $W>0$. Then,
\begin{align}\label{trmc3}
 &(\hat\Sigma_0,D_0)=\arg\min\{ \trace(W\hat\Sigma) \mid \Sigma=\hat\Sigma+D>0,\;\hat\Sigma\geq
 0,\;\mbox{diagonal }D\geq 0\}\\
&  \Leftrightarrow~ \exists~ \Lambda_0 \geq0 ~:~
\hat\Sigma \Lambda_0=0 \text{~and~} \left\{
                         \begin{array}{ll}
                           [\Lambda_0]_{ii}=[W]_{ii}, & \text{~if~} [D_0]_{ii}>0, \\
                           \left[ \Lambda_0\right]_{ii}\geq [W]_{ii}, &\text{~if~}[D_0]_{ii}=0.\nonumber
                         \end{array}
                       \right.\nonumber
\end{align}
\end{prop}

A corresponding sufficient and necessary condition for $(\hat\Sigma, D)$ to be a minimizer in Shapiro's problem is that there exists a Grammian in the null space of $\hat\Sigma$ whose diagonal entries are equal to the diagonal entries of $W$.

Minimum-rank solutions may be recovered as solutions to \eqref{trmc3} using suitable choices of weight.
However, these choices depend on $\Sigma$ and are not known in advance --this motivates a selection of certain canonical $\Sigma$-dependent
weight as well as iteratively improving the choice of weight. One should note that since $D$ is diagonal, letting $W$ be a not-necessarily
diagonal matrix does not change the problem --only the diagonal entries of $W$ determine the minimizer.

We first consider taking $W=\Sigma^{-1}$. A rationale for this
choice is that the minimal value in \eqref{eq:Dw} bounds
$\mr_+(\Sigma)$ from below, since for any decomposition
$\Sigma=\hat\Sigma+D$,
\begin{align}\nonumber
\rank(\hat \Sigma) =&~ \trace (\hat\Sigma^\sharp \hat\Sigma)\\\nonumber
\geq&~ \trace((\hat\Sigma+D)^{-1} \hat\Sigma)\\
=&~ \trace(\Sigma^{-1} \hat\Sigma)\label{eq:rankjustify}
\end{align}
where $^\sharp$ denotes the Moore-Penrose pseudo inverse. Continuing
with this line of analysis
\begin{align}
\rank(\hat \Sigma) =&~ \trace (\hat\Sigma^\sharp \hat\Sigma)\nonumber\\
\geq&~ \trace((\hat\Sigma+\epsilon I)^{-1} \hat\Sigma)\label{eq:ranktrace}
\end{align}
for any $\epsilon>0$, suggests the iterative re-weighting process
\begin{align}\label{minimizer_a}
D_{(k+1)}:=&~\arg\min_{D}\trace\left((\Sigma-D_{(k)} +\epsilon I)^{-1}(\Sigma-D)\right)
\end{align}
for $k=1,\,2,\,\ldots$ and $D_{(0)}:=0$.
In fact, as pointed out in \cite{Fazel2003}, \eqref{minimizer_a} corresponds to minimizing
$\log\det(\Sigma-D+\epsilon I)$
by local linearization.

Next we provide a sufficient condition for $\hat\Sigma$ to be such a
stationary point \eqref{minimizer_a}, i.e., 
for $\hat\Sigma$ to satisfy
\begin{align}\label{stationary_a}
\arg\min_{D}\trace\left((\hat\Sigma+\epsilon I)^{-1}(\hat\Sigma-D)\right)=0.
\end{align}
The notation $\circ$ used below denotes the
element-wise product between vectors or matrices which is also known
as \emph{Schur product} \cite{Horn1990matrix} and, likewise, for
vectors $a, b \in \mR^{n\times 1}$, $a\circ b\in \mR^{n\times 1}$
with $[a\circ b]_i=[a]_i[b]_i$.

\begin{prop}\label{prop:stationary_a}
Let $\hat\Sigma\in\Sym_{n,+}$ and let the columns of $U$ form a basis of $\cR(\hat\Sigma)$. If
\begin{align}\label{eq:stationary_a}
\cR(U\circ U) \subset \cR(\Pi_{\cN(\hat\Sigma)}\circ\Pi_{\cN(\hat\Sigma)} ),
\end{align}
then $\hat\Sigma$ satisfies \eqref{stationary_a} for all $\epsilon\in(0,\; \epsilon_1)$ and some $\epsilon_1>0$.
\end{prop}

We first need the following result which generalizes  \cite[Theorem 3.1]{Shapiro1985}.
\begin{lemma}\label{lemma:trace}
For $A\in \mR^{n\times p}$ and $B\in \mR^{n\times q}$ having columns $a_1, \ldots, a_p$ and $b_1, \ldots, b_q$, respectively, we let
\begin{align*}
C&=[a_1\circ b_1, a_1\circ b_2, \ldots,a_2\circ b_1\dots a_p \circ b_q]\in \mR^{n\times pq},\\
\phi &: ~\mR^n \hspace*{.4cm}\rightarrow \mR^n \hspace*{.57cm} d \mapsto \diag(AA'\diag^*(d)BB'), \mbox{ and}\\
\psi &: ~\mR^{p\times q} \rightarrow \mR^n \hspace*{.5cm}  \Delta \mapsto \diag(A\Delta B').
\end{align*}
Then $\cR(\phi)=\cR(\psi)=\cR((AA')\circ (BB'))=\cR(C)$.
\end{lemma}

\begin{proof}
Since $\diag(AA'\diag^*(d)BB')=((AA')\circ (BB'))d$, it follows that
\[\cR(\phi)=\cR((AA')\circ (BB').\]
Moreover,
$\diag(A\Delta B')= \sum_{i=1}^p\sum_{j=1}^q a_i\circ b_j [\Delta]_{ij}$, and then
$\cR(\psi)=\cR(C)$.
We only need to show that $\cR(C)=\cR((AA')\circ (BB'))$. This follows from
\begin{align*}
(AA')\circ (BB')       =&~\sum_{i=1}^p\sum_{j=1}^q (a_ia_i')\circ (b_jb_j')\\
         =&~\sum_{i=1}^p\sum_{j=1}^q (a_i\circ b_j) (a_i\circ b_j)'
         =CC'.
\end{align*}
Thus $\cR(C)=\cR((AA')\circ (BB'))$.
\end{proof}

\begin{proof} {\em [Proof of Proposition \ref{prop:stationary_a}:]}
Assume that $\hat\Sigma$ satisfies \eqref{stationary_a}.
If $\rank(\hat\Sigma)=r$, let $\hat\Sigma=USU'$ be the eigendecomposition of $\hat\Sigma$ with $S=\diag^*(s)$ with $s\in \mR^r$. Let the columns of $V$ be an orthogonal basis of the null space of $\hat\Sigma$, i.e., $\Pi_{\cN(\hat\Sigma)}=VV'$.
Then
\begin{align*}
(\hat\Sigma+\epsilon I)^{-1}=(\hat\Sigma+\epsilon \Pi_{\cR(\hat\Sigma)}+\epsilon  \Pi_{\cN(\hat\Sigma)})^{-1} =(\hat\Sigma+\epsilon \Pi_{\cR(\hat\Sigma)})^\sharp+\frac{1}{\epsilon} \Pi_{\cN(\hat\Sigma)},
\end{align*}
and
\begin{align*}
\arg\min_{D:\hat\Sigma\geq D} \trace\left((\hat\Sigma+\epsilon I)^{-1}(\hat\Sigma-D)\right)& =\\
&\hspace*{-1.5cm}\arg\min_{D:\hat\Sigma\geq D} \trace\left(\left(\epsilon( \hat\Sigma+\epsilon \Pi_{\cR(\hat\Sigma)})^\sharp+\Pi_{\cN(\hat\Sigma)}\right)(\hat\Sigma-D)\right).
\end{align*}
From Proposition \ref{prop:WMTFAShapiro}, \eqref{stationary_a} holds if there is $M\in \Sym_{r,+}$ such that
\begin{align}\label{stationaryaDiag}
\diag(VMV')=\diag\left( \epsilon( \hat\Sigma+\epsilon \Pi_{\cR(\hat\Sigma)})^\sharp+\Pi_{\cN(\hat\Sigma)}\right).
\end{align}
Obviously, if $\epsilon=0$
 $M=I$ satisfies the above equation. We consider the matrix $M$ of the form $M=I+\Delta$. For \eqref{stationaryaDiag} holds, we need $\diag((\hat\Sigma+\epsilon \Pi_{\cR})^\sharp)$ to be in the range of $\psi$ for
\[
\psi: \Sym_n \rightarrow \mR^n \hspace*{.57cm} \Delta \mapsto \diag(V\Delta V').
\]
From Lemma \ref{lemma:trace} that $\cR(\psi)=\cR(\Pi_{\cN(\hat\Sigma)}\circ\Pi_{\cN(\hat\Sigma)})$. On the other hand, since
\[
\epsilon(\hat\Sigma+\epsilon \Pi_{\cR(\hat\Sigma)})^\sharp=U\diag\left(\left[\frac{\epsilon}{[s]_1+\epsilon}, \ldots, \frac{\epsilon}{[s]_r+\epsilon} \right]\right)U',
\]
then $\diag(\epsilon(\hat\Sigma+\epsilon
\Pi_{\cR(\hat\Sigma)})^\sharp)\in \cR(U\circ U)$. So if
\eqref{eq:stationary_a} holds, there is always a $\Delta$ such that
$M=I+\Delta$ satisfies \eqref{stationaryaDiag}. Morover, it is also
required that $I+\Delta\geq0$. Since the map from $\epsilon$ to
$\Delta$ is continuous, for small enough $\epsilon$, i.e. in a
interval $(0, \epsilon_1)$ the condition $I+\Delta$ can always be
satisfied.
\end{proof}

We note that \eqref{eq:stationary_a} is a sufficient condition for $\hat\Sigma$ to be a stationary point of \eqref{stationary_a} in both Frisch's and Shapiro's settings.

\section{Certificates of minimum rank}\label{sec:CertifMinRank}

We are interested in obtaining bounds on the minimal rank for the
Frisch problem so as to ensure optimality when candidate solutions
are obtained by the earlier optimization approach in \eqref{minimizer_a}.

The following two bounds were proposed in \cite{Woodgate1}, and
follow from Theorem~\ref{thm:Reiersol}. However, both of these
bounds require exhaustive search which may be prohibitively
expensive when $n$ is large.
\begin{subequations}
\begin{cor}\label{cor1}
Let $\Sigma\in\Sym_{n,+}$ and $\Sigma>0.$ If there is an $s_1\times s_1$
principle minor of $\Sigma$ whose inverse is positive, then
 \begin{align}
\mr_+(\Sigma)&\geq s_1-1.
 \end{align}
 If there is an $s_2\times s_2$ principle
minor of $\Sigma^{-1}$ which is element-wise positive, then
 \begin{align}
\mr_+(\Sigma)&\geq s_2-1.
 \end{align}
\end{cor}

Next we discuss three other bounds that are computationally
more tractable --the first two were proposed by Guttman
\cite{Guttman1954}.
Guttman's bounds are based on a conservative assessment for the admissible range of each of the diagonal entries of $D=\Sigma-\hat\Sigma$.

\begin{prop}\label{prop:Guttman}
 Let $\Sigma\in \Sym_{n,+}$ and let
 \begin{align*}
 D_1&:=\diag^*(\diag(\Sigma))\\
 D_2&:=\left(\diag^*(\diag(\Sigma^{-1}))\right)^{-1}.
 \end{align*}
 Then the following hold,
\begin{align}
&\mr_+(\Sigma)\geq n_+(\Sigma-D_1) \label{Guttman:bound1}\\
&\mr_+(\Sigma)\geq n_+(\Sigma-D_2). \label{Guttman:bound2}\\
\nonumber
\end{align}
Further, $n_+(\Sigma-D_1)\leq n_+(\Sigma-D_2)$.
\end{prop}

\begin{proof} The proof follows from the fact that $\Sigma\geq D$ implies $D\leq D_2\leq D_1$. See \cite{Guttman1954} for details.
\end{proof}

It is also easy to see that $\mr(\Sigma)\geq n_+(\Sigma-D_1)$ which
provides a lower bound for the minimum rank in Shapiro's problem.
Next we return to a bound, which we noted earlier in \eqref{eq:rankjustify}.

\begin{prop}\label{tracebound}
Let $\Sigma\in\Sym_{n,+}$. Then the following holds:
\begin{align}\label{eq:lowerbound3}
\mr_+(\Sigma)\geq \min_{\Sigma\geq D\geq 0}\trace(\Sigma^{-1}(\Sigma-D)).
\end{align}
\end{prop}

\begin{proof} The statement follows readily from \eqref{eq:rankjustify}.
\end{proof}

Evidently an analogous statement holds for $\mr(\Sigma)$.
We note that \eqref{Guttman:bound1} and \eqref{Guttman:bound2} remain invariant
under scaling of rows and corresponding columns, whereas \eqref{eq:lowerbound3} does not, hence these two cannot be compared directly.
\end{subequations}

\section{Correspondence between decompositions}\label{correspondence}

We now return to the decomposition of the data matrix $X=\hat
X+\tilde X$ as in \eqref{eq:decompose} and its relation to the
corresponding sample covariances. The decomposition of $X$ into
``noise-free'' and ``noisy'' components implies a corresponding
decomposition for the sample covariance, but in the converse
direction, a decomposition $ \Sigma=\hat\Sigma+\tilde\Sigma $ leads
to a family of compatible decompositions for $X$, which corresponds
to the boundary of a matrix-ball. This is discussed next.

\begin{prop}\label{prop:decomposition} Let $X\in\mR^{n\times T}$, and $\Sigma:=XX^\prime$. If
\begin{equation}\label{eq:decompose2}
\Sigma=\hat \Sigma+\tilde \Sigma
\end{equation}
with $\hat \Sigma$,  $\tilde \Sigma$ symmetric and non-negative definite, there exists a decomposition
\begin{subequations}\label{conditions}
\begin{equation}\label{eq:Xdecompose}
X=\hat X+\tilde X
\end{equation}
for which
\begin{eqnarray}
\label{cond2}
&&\hat X \tilde X^\prime  = 0,\\\label{cond3}
&&\hat \Sigma = \hat X \hat X^\prime,\\\label{cond4}
&&\tilde \Sigma = \tilde X\tilde X^\prime.
\end{eqnarray}
\end{subequations}
Further,
all pairs $(\hat X,\,\tilde X)$ that satisfy (\ref{eq:Xdecompose}-\ref{cond4})
are of the form
\begin{equation}\label{parametrization}
\hat X=\hat\Sigma \Sigma^{-1} X+R^{1/2}V,\;
\tilde X=\tilde\Sigma \Sigma^{-1} X-R^{1/2}V,
\end{equation}
with
\begin{subequations}\label{Rs}
\begin{eqnarray}\label{R}
R&:=&\hat \Sigma - \hat \Sigma \Sigma^{-1} \hat \Sigma\\
&=&\tilde \Sigma - \tilde \Sigma \Sigma^{-1} \tilde \Sigma \label{R2}\\\nonumber
&=&\hat \Sigma \Sigma^{-1}\tilde \Sigma\\\nonumber
&=&\tilde \Sigma \Sigma^{-1}\hat \Sigma,\nonumber
\end{eqnarray}
\end{subequations}
and $V\in\mR^{n\times T}$ such that $VV'=I$, $XV'=0$.
\end{prop}

\begin{proof} The proof relies on a standard lemma (\cite[Theorem 2]{douglas}) which states that if
 $A\in\mR^{n\times T}$, $B\in\mR^{n\times m}$ with $m\leq T$ such that
$A A^\prime = B B^\prime,$
then $A=BU$ for some $U\in\mR^{m\times T}$ with $U U^\prime =I$.
Thus, we let $A:=X$,
\[
S:=\left[\begin{matrix}\hat \Sigma & 0\\0&\tilde \Sigma\end{matrix}\right],
\]
and $B:=\left[\begin{matrix}I&I \end{matrix}\right] S^{1/2}$,
where $S^{1/2}$ is the matrix-square root of $S$.
It follows that there exists a matrix $U$ as above for which $A=BU$, and therefore we can take
\[
\left[\begin{matrix}\hat X \\ \tilde X\end{matrix}\right]:=S^{1/2}U.
\]
This establishes the existence of the decomposition
\eqref{eq:Xdecompose}.

In order to parameterize all such pairs $(\hat X,\,\tilde X)$, let
$U_o$ be an orthogonal (square) matrix such that
\[XU_o=[\Sigma^{1/2} \; 0].
\]
Then $\hat X U_o$ and $\tilde X U_o$ must be of the form
\begin{equation}\label{UXs}
\hat X U_o=: \left[\begin{matrix}\hat X_1&\Delta \end{matrix}\right],\;
\tilde X U_o=: \left[\begin{matrix}\tilde X_1& -\Delta \end{matrix}\right],
\end{equation}
with $\hat X_1$, $\tilde X_1$ square matrices. Since
\[\left[\begin{matrix}\hat X\\\tilde X \end{matrix}\right]
\left[\begin{matrix}\hat X^\prime &\tilde X^\prime  \end{matrix}\right]
=\left[\begin{matrix}\hat \Sigma& 0\\0&\tilde \Sigma \end{matrix}\right],
\]
then
\begin{subequations}
\begin{eqnarray}\label{first}
&&\hat X_1\hat X_1^\prime+\Delta\Delta^\prime=\hat \Sigma\\\label{second}
&&\hat X_1\tilde X_1^\prime-\Delta\Delta^\prime=0\\\label{third}
&&\tilde X_1\tilde X_1^\prime+\Delta\Delta^\prime=\tilde \Sigma.
\end{eqnarray}
\end{subequations}
Substituting $\hat X_1\tilde X_1^\prime$ for $\Delta\Delta^\prime$ into (\ref{first}) and
using the fact that $\tilde X_1=X_1-\hat X_1$ with $X_1=\Sigma^{1/2}$ we obtain that
\begin{eqnarray*}
&&\hat X_1=\hat \Sigma\Sigma^{-1/2}.
\end{eqnarray*}
Similarly, using (\ref{third}) instead, we obtain that
\begin{eqnarray*}
&&\tilde X_1=\tilde \Sigma\Sigma^{-1/2}.
\end{eqnarray*}
Substituting into (\ref{second}), (\ref{first}) and (\ref{third}) we obtain the following three relations
\begin{eqnarray*}
\Delta\Delta^\prime &=& \hat \Sigma \Sigma^{-1}\tilde \Sigma\\
&=&\hat \Sigma - \hat \Sigma \Sigma^{-1} \hat \Sigma\\
&=&\tilde \Sigma - \tilde \Sigma \Sigma^{-1} \tilde \Sigma.
\end{eqnarray*}
Since $\Delta\Delta^\prime$ and the $\Sigma$'s  are all symmetric,
\begin{eqnarray*}
\Delta\Delta^\prime&=&\tilde \Sigma \Sigma^{-1}\hat \Sigma
\end{eqnarray*}
as well. Thus, $\Delta=R^{1/2}V_1$ with $V_1V_1^\prime=I$. The proof
is completed by substituting the expressions for $\hat X_1$ and
$\Delta$ into \eqref{UXs}.
\end{proof}

Interestingly,
\[
\rank(R)+\rank(\Sigma)=\rank \left(\left[
                               \begin{array}{cc}
                                 \hat\Sigma & \hat\Sigma \\
                                 \hat\Sigma & \Sigma \\
                               \end{array}
                             \right] \right)=\rank \left(\left[
                               \begin{array}{cc}
                                 \hat\Sigma & 0 \\
                                 0 & \tilde\Sigma \\
                               \end{array}
                             \right] \right)=\rank(\hat\Sigma)+\rank(\tilde\Sigma),
\]
and hence, the rank of the ``uncertainty radius'' $R$ of the corresponding $\hat X$ and $\tilde X$-matrix spheres is
\[\rank(R)= \rank(\hat\Sigma)+\rank(\tilde\Sigma)-\rank(\Sigma).
\]
In cases where identifying $\hat X$ from the data matrix $X$,
different criteria may be used to quantify uncertainty. One such is
the rank of $R$ while another is its trace, which is the variance of
estimation error in determining $\hat X$. This topic is considered
next and its relation to the Frisch decomposition highlighted.

\section{Uncertainty and worst-case estimation}\label{sec:min-max}
The basic premise of the decomposition (\ref{eq:decompose2}) is
that, in principle, no probabilistic description of the data is
needed. Thus, under the assumptions of
Proposition~\ref{prop:decomposition}, $R$ represents a deterministic
radius of uncertainty in interpreting the data. On the other hand,
when data and noise are probabilistic in nature and represent
samples of jointly Gaussian random vectors $\bx,\;\bhx,\; \btx$ as
in (\ref{eq:xa} - \ref{eq:xc}),  the conditional expectation of
$\bhx$ given $\bx$ is $E\{\bhx |\bx\}=\hat\Sigma\Sigma^{-1} \bx$,
while the variance of the error
\begin{eqnarray*}
E\{(\bhx-\hat\Sigma \Sigma^{-1}\bx)(\bhx-\hat\Sigma \Sigma^{-1}\bx)^\prime\}&=&\hat\Sigma - \hat\Sigma\Sigma^{-1}\hat\Sigma\\
&=&R
\end{eqnarray*}
is the radius of the deterministic uncertainty set. Either way, it is of interest to assess how this radius depends on the decomposition of $\Sigma$.

\subsection{Uniformly optimal decomposition}

Since the decomposition of $\Sigma$ in the Frisch problem is not
unique, it is natural to seek a uniformly optimal choice of the
estimate $K\bx$ for $\bhx$ over all admissible decompositions. To
this end, we denote the mean-squared-error loss function
\begin{eqnarray}\label{eq:LossFunction}
L(K, \hat \Sigma, \tilde\Sigma)&:=&\trace\left(\cE\left( (\bhx-K\bx)(\bhx-K\bx)^\prime\right)\right)\nonumber\\
&\;=&\trace\left(\hat\Sigma-K\hat\Sigma-\hat\Sigma K'+K(\hat\Sigma+\tilde\Sigma) K' \right),\label{eq:loss}
\end{eqnarray}
and define
\begin{align*}
\cS(\Sigma):= \{(\hat\Sigma, \tilde\Sigma) : &~\Sigma=\hat\Sigma+\tilde\Sigma,\; \hat\Sigma,\; \tilde\Sigma\geq0 \text{~and~} \tilde\Sigma \text{~is diagonal} \}
\end{align*}
as the set of all admissible pairs.  Thus, a uniformly-optimal
decomposition of $X$ into signal plus noise relates to the following
min-max problem:
\begin{align}\label{prob:minmax}
\min_{K}\max_{(\hat\Sigma,\tilde\Sigma)\in\cS(\Sigma)}  L(K, \hat \Sigma, \tilde\Sigma).
\end{align}
The minimizer of \eqref{prob:minmax} is the uniformly optimal
estimator gain $K$. Analogous min-max problems, over
different uncertainty sets, have been studied in the literature
\cite{Eldar2004competitive}. In our setting
\begin{subequations}\label{eq:concave}
\begin{eqnarray}
\min_{K}\max_{(\hat\Sigma,\tilde\Sigma)\in\cS(\Sigma)}  L(K, \hat \Sigma,\tilde\Sigma)&\geq&\max_{(\hat\Sigma,\tilde\Sigma)\in\cS(\Sigma)}\min_{K}  L(K, \hat \Sigma,\tilde\Sigma)\label{minmaxmaxmin}\\
&=&\max_{(\hat\Sigma,\tilde\Sigma)\in\cS(\Sigma)} \trace\left(\hat\Sigma-\hat\Sigma\Sigma^{-1}\hat\Sigma\right)\label{eq:concave1}\\
&=&\max_{(\hat\Sigma,\tilde\Sigma)\in\cS(\Sigma)} \trace\left(\tilde\Sigma-\tilde\Sigma\Sigma^{-1}\tilde\Sigma\right).\label{eq:concave2}
\end{eqnarray}
\end{subequations}
The functions to maximize in \eqref{eq:concave1} and
\eqref{eq:concave2} are both strictly concave in $\hat\Sigma$ and
$\tilde\Sigma$. Therefore the maximizer is unique. Thus, we denote
\begin{equation}\label{optsolution}
(K_{\rm opt}, \hat\Sigma_{\rm opt}, \tilde\Sigma_{\rm opt}) :=\arg \max_{(\hat\Sigma,\tilde\Sigma)\in\cS(\Sigma)}\min_{K}  L(K, \hat \Sigma,\tilde\Sigma),
\end{equation}
where, clearly, $K_{\rm opt}=\hat\Sigma_{\rm opt}\Sigma^{-1}$.

In general, the decomposition suggested by the uniformly optimal
estimation problem does not lead to a singular signal covariance
$\hat\Sigma$. The condition for when that happens is given next.
Interestingly, this is expressed in terms of half the candidate
noise covariance utilized in obtaining one of the Guttman bounds
(Proposition \ref{prop:Guttman}).

\begin{prop}\label{prop:maxmin}
Let $\Sigma>0$, and let
\begin{equation}\label{eq:D0}
D_0:=\frac12 \diag^*\left(\diag(\Sigma^{-1})\right)^{-1}
\end{equation}
(which is equal to $\frac12 D_2$ defined in Proposition \ref{prop:Guttman}).
If $\Sigma-D_0\geq0$, then
\begin{subequations}\label{Solution}
\begin{equation}\label{InteriorSolution}
\tilde\Sigma_{\rm opt}=D_0 \text{~and~} \hat\Sigma_{\rm opt}=\Sigma-D_0.
\end{equation}
Otherwise,
\begin{equation}\label{BoundarySolution}
\tilde\Sigma_{\rm opt}\leq D_0 \text{~and~} \hat\Sigma_{\rm opt} \text{~is singular}.
\end{equation}
\end{subequations}
\end{prop}

\begin{proof}
From \eqref{eq:concave2},
\begin{eqnarray}
L(K_{\rm opt}, \hat \Sigma_{\rm opt}, \tilde \Sigma_{\rm opt})&=&\max \left\{\tilde \Sigma-\tilde \Sigma\Sigma^{-1}\tilde \Sigma ~\mid~ \Sigma\geq\tilde\Sigma\geq0, \tilde\Sigma \text{~is diagonal}  \right\}\nonumber\\
&\leq& \max \left\{\tilde \Sigma-\tilde \Sigma\Sigma^{-1}\tilde \Sigma ~\mid~\tilde\Sigma \text{~is diagonal}  \right\}\label{relaxedD}\\
&=&\frac12 \trace(D_0)\nonumber
\end{eqnarray}
with the maximum attained for $\tilde\Sigma=D_0$. Then
\eqref{InteriorSolution} follows. In order to prove
\eqref{BoundarySolution}, consider the Lagrangian corresponding to
\eqref{eq:concave2}
\[
\cL(\tilde\Sigma,\Lambda_0, \Lambda_1) =\trace(\tilde\Sigma-\tilde\Sigma\Sigma^{-1}\tilde\Sigma+\Lambda_0(\Sigma-\tilde\Sigma)+\Lambda_1\tilde\Sigma)
\]
where $\Lambda_0,\;\Lambda_1$ are Lagrange multipliers.
The optimal values satisfy
\begin{subequations}
\begin{eqnarray}
&&[I-2\Sigma^{-1}\tilde\Sigma_{\rm opt}-\Lambda_{0}+\Lambda_{1}]_{kk}=0, \;\forall\; k=1,\ldots, n,\label{condition1}\\
&& \Lambda_{0}\hat\Sigma_{\rm opt}=0,\; \Lambda_{0}\geq0,\label{condition2}\\
&& \Lambda_{1}\tilde\Sigma_{\rm opt}=0,\; \Lambda_{1}\geq0 \text{~and is diagonal}.\label{condition3}
\end{eqnarray}
\end{subequations}
If $\Sigma- D_0\not\geq0$ we show that $\hat\Sigma_{\rm opt}$ is
singular. Assume the contrary, i.e., that $\hat\Sigma_{\rm opt}>0$.
From \eqref{condition2}, we see that $\Lambda_{0}=0$, while from
\eqref{condition1}, $ [I-2\Sigma^{-1}\tilde\Sigma_{\rm
opt}]_{kk}\leq 0. $ This gives that
\[
[\tilde\Sigma_{\rm opt}]_{kk}\geq \frac{1}{2[\Sigma^{-1}]_{kk}}= [D_0]_{kk},
\]
for all $k=1, \ldots, n$, which contradicts the fact that $\Sigma-D_0\not\geq0$. Therefore $\hat\Sigma_{\rm opt}$ is singular.
We now assume that $\tilde\Sigma\not \leq D_0$. Then there exists $k$ such that $[\tilde\Sigma_{\rm opt}]_{kk}> [D_0]_{kk}$.
From \eqref{condition3} and \eqref{condition1}, we have that
 \[
 [\Lambda_{1}]_{kk}=0 \text{~and~} [I-2\Sigma^{-1}\tilde\Sigma_{\rm opt}]_{kk}\geq0
 \]
which contradicts the assumption that $[\tilde\Sigma_{\rm opt}]_{kk}>
[D_0]_{kk}$. Therefore $\tilde\Sigma_{\rm opt}\leq D_0$ and
\eqref{BoundarySolution} has been established.
\end{proof}

We remark that while
\begin{eqnarray*}
\cE\left( (\bhx-K\bx)(\bhx-K\bx)^\prime\right)&=&\hat\Sigma-K\hat\Sigma-\hat\Sigma K'+K\Sigma K'\\
&=&(\hat\Sigma\Sigma^{-\frac12}-K\Sigma^{\frac12})(\hat\Sigma\Sigma^{-\frac12}-K\Sigma^{\frac12})^\prime+\hat\Sigma-\hat\Sigma\Sigma^{-1}\hat\Sigma
\end{eqnarray*}
is matrix-convex in $K$ and a unique minimum for
$K=\hat\Sigma\Sigma^{-1}$, the error covariance
$\hat\Sigma-\hat\Sigma\Sigma^{-1}\hat\Sigma $ may not have a unique
maximum in the positive semi-definite sense. To see this, consider
$\Sigma=\left[
                   \begin{array}{cc}
                     2 & 1 \\
                     1 & 2 \\
                   \end{array}
                 \right]
$. In this case $D_0=\frac{3}{4}I$, $\hat\Sigma_{\rm opt}=\left[
                   \begin{array}{cc}
                     5/4 & 1 \\
                     1 & 5/4 \\
                   \end{array}
                 \right]$, and
\begin{equation}\label{eq:Ropt}
\hat\Sigma_{\rm opt}-\hat\Sigma_{\rm opt}\Sigma^{-1}\hat\Sigma_{\rm opt}=\left[
                   \begin{array}{cc}
                     3/8 & 3/16 \\
                     3/16 & 3/8 \\
                   \end{array}
                 \right].
\end{equation}
On the other hand, for $\hat\Sigma=\left[
                   \begin{array}{cc}
                     3/2 & 1 \\
                     1 & 3/2 \\
                   \end{array}
                 \right]$, then
\[
\hat\Sigma-\hat\Sigma\Sigma^{-1}\hat\Sigma=\left[
                   \begin{array}{cc}
                     1/3 & 1/12 \\
                     1/12 & 1/3 \\
                   \end{array}
                 \right]
\]
which is neither larger nor smaller than \eqref{eq:Ropt} in the
sense of semi-definiteness. This is a key reason for considering
scalar loss functions of the error covariance as in
\eqref{eq:loss}.

Next we note that there is no gap between the min-max and max-min
values in the two sides of \eqref{minmaxmaxmin}.
\begin{prop}\label{prop:minmax}
For $\Sigma\in\Sym_{n,+}$, then
\begin{equation}\label{eq:equal}
\min_{K}\max_{(\hat\Sigma, \tilde\Sigma)\in\cS(\Sigma)}  L(K, \hat \Sigma, \tilde\Sigma)=\max_{(\hat\Sigma, \tilde\Sigma)\in\cS(\Sigma)}\min_{K}  L(K, \hat \Sigma, \tilde\Sigma).
\end{equation}
\end{prop}
\begin{proof}
We observe that for a fixed $K$, the function $L(K, \hat \Sigma,
\tilde\Sigma)$ is a linear function of $(\hat\Sigma, \tilde\Sigma)$.
For fixed $(\hat\Sigma, \tilde\Sigma)$, the function is a convex
function of $K$. Under this conditions it is standard that
\eqref{eq:equal} holds, see e.g. \cite[page 281]{Boyd2004convex}.
\end{proof}

We remark that when $D_0=\frac12
\diag^*\left(\diag(\Sigma^{-1})\right)^{-1}$ is admissible as noise
covariance, i.e., $\Sigma- D_0\geq0$, the optimal signal covariance
is $\hat\Sigma_{\rm opt}=\Sigma-D_0$, and the gain matrix $K_{\rm
opt}=\hat\Sigma_{\rm opt}\Sigma^{-1}=I-D_0\Sigma^{-1}$ has all
diagonal entries equal to $\frac{1}{2}$. Thus, with $K_{\rm opt}$ in
\eqref{eq:LossFunction} the mean-square-error loss is independent of
$\hat\Sigma$ and equal to $\trace\left(K_{\rm opt}\Sigma K_{\rm
opt}^\prime\right)$ for any admissible decomposition of $\Sigma$.

We also remark that the key condition (Proposition \ref{prop:maxmin})
 \begin{align*}\label{InvariantCondition}
& \Sigma\geq\frac12 \diag^*\left(\diag(\Sigma^{-1}) \right)^{-1}\\
&\Leftrightarrow 2\diag^*\left(\diag(\Sigma^{-1}) \right)\geq \Sigma^{-1}
 \end{align*}
can be equivalently written as $\Sigma^{-1}\circ (2I-{\bf 1}{\bf
1}')\geq0$, and interestingly, amounts to the positive
semi-definitess of a matrix formed by changing the signs of all
off-diagonal entries of $\Sigma^{-1}$. The set of all such matrices, $\left\{S
\mid S\geq 0,~ S\circ (2I-{\bf 1}{\bf 1}')\geq0 \right\}$, is
convex, invariant under scaling rows and corresponding columns, and
contains the set of diagonally dominant matrices $\{S \mid S\geq 0,~
[S]_{ii}\geq \sum_{j\neq i} |[S]_{ij}| \text{~for all ~} i\}$.

We conclude this section by noting that $\trace(R_{\rm opt})$,
 with
 \[
R_{\rm opt}:=\hat\Sigma_{\rm opt}-\hat\Sigma_{\rm opt}\Sigma^{-1}\hat\Sigma_{\rm opt},
\]
quantifies the distance between admissible decompositions of $\Sigma$. This is stated next.

\begin{prop}\label{lemma:radius}
For $\Sigma>0$ and any pair $(\hat\Sigma, \tilde\Sigma)\in \cS(\Sigma)$,
\[
\trace\left( (\hat\Sigma-\hat\Sigma_{\rm opt})\Sigma^{-1}(\hat\Sigma-\hat\Sigma_{\rm opt})' \right)\leq \trace(R_{\rm opt}).
\]
\end{prop}
\begin{proof}
Clearly
$0\leq\trace(\hat\Sigma-\hat\Sigma\Sigma^{-1}\hat\Sigma)$,
while from Proposition \ref{prop:minmax},
\begin{eqnarray}
L(K_{\rm opt}, \hat\Sigma, \tilde\Sigma)
&=& \trace(\hat\Sigma-2\hat\Sigma_{\rm opt}\Sigma^{-1}\hat\Sigma+\hat\Sigma_{\rm opt}\Sigma^{-1}\hat\Sigma_{\rm opt}')\label{eqB}\\
&\leq& \trace(R_{\rm opt}).\nonumber
\end{eqnarray}
Thus, $\trace(\hat\Sigma\Sigma^{-1}\hat\Sigma-2\hat\Sigma_{\rm opt}\Sigma^{-1}\hat\Sigma+\hat\Sigma_{\rm opt}\Sigma^{-1}\hat\Sigma_{\rm opt}')\leq \trace(R_{\rm opt})$.
\end{proof}

\subsection{Uniformly optimal estimation and trace regularization}\label{sec:regularized}
A decomposition of $\Sigma$ in accordance with the min-max estimation problem of the previous section often produces an invertible signal covariance $\hat\Sigma$. On the other hand, it is often the case and it is the premise of factor analysis, that $\hat\Sigma$ is singular of low rank and, thereby, allows identifying linear relations in the data. In this section we consider combining the mean-square-error loss function with regularization term promoting a low rank for the signal covariance $\hat\Sigma$ \cite{Fazel2001}. More specifically, we consider
\begin{equation}\label{prob:minmaxRank}
J=\min_{K}\max_{(\hat\Sigma,\tilde\Sigma)\in\cS(\Sigma)}  \left(L(K,
\hat \Sigma, \tilde\Sigma)-\lambda\cdot \trace(\hat\Sigma)\right),
\end{equation}
for $\lambda\geq0$, and properties of its solutions.

As noted in Proposition \ref{prop:minmax} (see \cite[page 281]{Boyd2004convex}), here too there is no gap between the min-max and the max-min, which becomes
\begin{subequations}
\begin{align}
&\max_{(\hat\Sigma,\tilde\Sigma)\in\cS(\Sigma)}\min_{K}  L(K, \hat \Sigma, \tilde\Sigma)-\lambda\cdot \trace(\hat\Sigma)\nonumber\\
&= \max_{(\hat\Sigma,\tilde\Sigma)\in\cS(\Sigma)}\min_{K}  \trace\left( (1-\lambda)\hat\Sigma-K\hat\Sigma-\hat\Sigma K'+K(\hat\Sigma+\tilde\Sigma)K'  \right)\nonumber\\
&=\max_{(\hat\Sigma,\tilde\Sigma)\in\cS(\Sigma)} \trace\left( (1-\lambda)\hat\Sigma-\hat\Sigma(\hat\Sigma+\tilde\Sigma)^{-1}\hat\Sigma \right) \label{eq:Kcanceled}\\
&=\max_{(\hat\Sigma,\tilde\Sigma)\in\cS(\Sigma)} \trace\left( -\lambda\Sigma+ (1+\lambda)\tilde\Sigma-\tilde\Sigma(\hat\Sigma+\tilde\Sigma)^{-1}\tilde\Sigma \right). \label{eq:Sigtilde}
\end{align}
\end{subequations}
Since \eqref{eq:Kcanceled}
and \eqref{eq:Sigtilde} are strictly concave functions of
$\hat\Sigma$ and $\tilde\Sigma$, respectively, there is a unique
set of optimal values $(K_{\lambda, \rm opt}, \hat\Sigma_{\lambda,\rm opt}, \tilde\Sigma_{\lambda,\rm opt})$.

\begin{prop}
Let $\Sigma>0$, $D_0=\frac12 \left(\diag^*\diag(\Sigma^{-1})\right)^{-1},$
$\lambda_{\rm min}$ be the smallest eigenvalue of $D_0^{-\frac12}\Sigma D_0^{-\frac12}$,
and $(K_{\lambda, \rm opt}, \hat\Sigma_{\lambda,\rm opt}, \tilde\Sigma_{\lambda,\rm opt})$ as above, for $\lambda\geq0$.
For any
$\lambda\geq\lambda_{\rm min}-1$,
$\hat\Sigma_{\lambda,{\rm opt}}$ is singular.
\end{prop}

\begin{proof}
The trace of
$( -\lambda\Sigma+ (1+\lambda)\tilde\Sigma-\tilde\Sigma\Sigma^{-1}\tilde\Sigma )$ is maximal for the diagonal choice $\tilde \Sigma = (1+\lambda)D_0$.
For any $\lambda \geq \lambda_{\rm min}-1$,  $\Sigma-(1+\lambda) D_0$ fails to be positive semidefinite. Thus, the constraint $\Sigma-\tilde\Sigma\geq 0$ in \eqref{eq:Sigtilde} is active and $\hat\Sigma_{\lambda, {\rm opt}}$ is singular.
\end{proof}

Note that $\Sigma-2D_0\not\geq 0$ (unless $\Sigma$ is diagonal), and therefore $\lambda_{\rm min}<2$. Hence, for
$\lambda\geq1$, $\hat\Sigma_{\lambda, {\rm opt}}$ is singular.
When $\lambda\to 0$ we recover the solution in \eqref{optsolution}, whereas for $\lambda\to\infty$ we recover the solution in Proposition
\ref{prop:mintrace}.

\section{Accounting for statistical errors}\label{statisticalerrors}

From an applications standpoint $\Sigma$ represents an empirical
covariance, estimated on the basis of a finite observation record in
$X$. Hence \eqref{eq:diagonal} and \eqref{eq:orthogonality} are only
approximately valid, as already suggested in
Section~\ref{sec:datastrcuture}. Thus, in order to account for
sampling errors we can introduce a penalty for the size of
 $C:=\hat X\tilde X^\prime$, conditioned so that
 \[
 \Sigma=\hat\Sigma + \tilde\Sigma +C +C',
 \]
and a penalty for the distance of $\tilde \Sigma$ from the set $\{D \mid D\mbox{ diagonal}\}$.

Alternatively, we can use the  Wasserstein 2-distance
\cite{olkin1982,ning2011} between the respective Gaussian
probability density functions, which can be written in the form of a
semidefinite program
\[
d(\hat\Sigma+D, \Sigma)=\min_{C_1}\left(\trace(\Sigma+\hat\Sigma+D+C_1+C_1') \mid \left[
\begin{array}{cc} \hat\Sigma+D & C_1 \\
                            C_1' & \Sigma \\
\end{array}
\right]\geq0 \right).
\]

Returning to the uncertainty radius of Section \ref{correspondence} and the
problem discussed in Section \ref{sec:min-max}, we note that the problem
\begin{equation}\nonumber
\max\min_{K}  L(K, \hat \Sigma,D)\\
=\max \trace\left(\hat\Sigma-\hat\Sigma(\hat\Sigma+D)^{-1}\hat\Sigma\right)
\end{equation}
can be expressed as the semidefinite program
\begin{equation}\nonumber
\max_Q \left\{ \trace\left(\hat\Sigma-Q \right)\mid
\left[
 \begin{array}{cc}
 Q & \hat\Sigma \\
 \hat\Sigma & \hat\Sigma+D \\
 \end{array}
 \right]\geq 0
\right\}.
\end{equation}
Thus, putting the above together, a formulation that incorporates the various tradeoffs between the dimension of the signal subspace, mean-square-error loss, and statistical errors is to maximize
\begin{equation}\label{eq:maxmin2}
\trace(\hat\Sigma -Q) - \lambda_1\, \trace(\hat\Sigma) -\lambda_2\, \trace(\hat\Sigma + D - C_1-C_1^\prime)
\end{equation}
subject to
\begin{eqnarray*}
\left[
 \begin{array}{cc}
 Q & \hat\Sigma \\
 \hat\Sigma & \hat\Sigma+D \\
 \end{array}
 \right]\geq 0,\;\left[
\begin{array}{cc} \hat\Sigma+D & C_1 \\
                            C_1' & \Sigma \\
\end{array}
\right]\geq0, \mbox{ with }D\geq 0 \mbox{ and diagonal.}
\end{eqnarray*}
The value of the parameters $\lambda_1$, $\lambda_2$ dictate the relative importance that we place on the various terms and determine the tradeoffs in the problem.

We conclude with an example to highlight the potential and limitations of the techniques.
We generate data $X$ in the form
\[
X=FV+\tilde X
\]
where $F\in\mR^{n\times r}$, $V\in \mR^{r\times T}$, and $\tilde X\in \mR^{n\times T}$ with $n=50$, $r=10$, $T=100$. The elements of $F$ and $V$ are generated from normal distributions with mean zero and unit covariance. The columns of $\tilde X$ are generated from a normal distribution with mean zero and diagonal covariance, itself having (diagonal) entries which are uniformly drawn from interval $[1, 10]$. The matrix $\Sigma=XX'$ is subsequently scaled so that $\trace(\Sigma)=1$.
We determine
\[
(\hat\Sigma,Q,D)={\rm arg}\max \left\{ \trace(\hat\Sigma-Q)-\lambda\cdot \trace(\hat\Sigma)\right\}
\]
subject to
\begin{eqnarray*}
\left[\begin{matrix}Q &\hat\Sigma\\ \hat\Sigma & \hat\Sigma+D \end{matrix} \right]\geq0, ~d(\hat\Sigma+D, \Sigma)\leq \epsilon, \text{~with~} \hat\Sigma, D\geq0 \text{~and~} D \text{~diagonal},
\end{eqnarray*}
and tabulate below a typical set of values for the rank of $\hat\Sigma$ (Table 1)
as a function of $\lambda$ and $\epsilon$. We observe a ``plateau'' where the rank stabilizes at $10$ over a small range of values for $\epsilon$ and $\lambda$. Naturally, such a plateau may be taken as an indication of a suitable range of parameters.
Although the current setting where a small perturbation in the empirical covariance $\Sigma$ is allowed, the bounds for the rank
in \eqref{Guttman:bound2} and \eqref{eq:lowerbound3} are still pertinent. In fact, for this example, in $7/10$ instances where the $\rank(\hat\Sigma)=10$ the bound in \eqref{Guttman:bound2} (computed based on the perturbed covariance $\hat\Sigma+D$) has been tight and it thus a valid certificate. For the same range of parameters, the bound in \eqref{eq:lowerbound3} has been lower than the actual rank of $\hat\Sigma$. In general, the bounds in \eqref{Guttman:bound2} and \eqref{eq:lowerbound3} are not comparable as either one may be tighter than the other.\\
\begin{center}
\begin{minipage}[]{.6\textwidth}
\begin{tabular}[t]{|c||c|c|c|c|c|c|c|}
  \hline
   	\backslashbox[.5cm]{$\lambda$}{$\epsilon$}	& $0$& $0.08$ & $0.10$& $0.12$& $0.14$ & $0.16$\\ \hline\hline
  $1$  & 46 & 26 & 24 & 23 & 22 & 22 \\ \hline
  $5$  & 46 & 17 & 14 & 10 & 10 & 9 \\ \hline
  $10$ & 45 & 16 & 12 & 10 & 10 & 8 \\ \hline
  $20$ & 45 & 15 & 12 & 10 & 10 & 8 \\  \hline
  $50$ & 45 & 15 & 12 & 10 & 10 & 8 \\  \hline
  $100$& 45 & 15 & 11 & 10 & 10 & 8 \\  \hline
\end{tabular}\\[.05in]
{Table 1: $\rank(\hat\Sigma)$ as a function of $\lambda$ and $\epsilon$}\\[.05in]
\end{minipage}
\end{center}

\section{Conclusions} \label{sec:conclusion}

In this paper we considered the general problem of identifying
linear relations among variables based on noisy measurements --a
classical problem of major importance in the current era of ``Big
Data.'' Novel numerical techniques and increasingly powerful
computers have made it possible to successfully treat a number of
key issues in this topic in a unified manner. Thus, the goal of the
paper has been to present and develop in a unified manner key ideas
of the theory of noise-in-variables linear modeling.

More specifically, we considered two different viewpoints for the
linear model problem under the assumption of independent noise. From
an estimation viewpoint, we quantify the uncertainty in estimating
``noise-free'' data based on noise-in-variables linear models. We
proposed a min-max estimation problem which aims at a uniformly
optimal estimator --the solution can be obtained using convex
optimization. From the modeling viewpoint, we also derived several
classical results for the Frisch problem that asks for the maximum
number of simultaneous linear relations. Our results provide a
geometric insight to the Reiers\o l theorem, a
 generalization to complex-valued matrices, an
iterative re-weighting trace minimization scheme for obtaining
solutions of low rank along with a characterization of fixed points,
and certain computational tractable lower bounds to serve as
certificates for identifying the minimum rank.  Finally, we consider
regularized min-max estimation problems which integrate various
objectives (low-rank, minimal worst-case estimation error) and
explain their effectiveness in a numerical example.

In recent years, techniques such as the ones presented in this work
are becoming increasingly important in subjects where one has very
large noisy datasets including medical imaging, genomics/proteomics,
and finance. It is our hope that the material we presented in this
paper will be used in these topics. It must be noted that throughout
the present work we emphasized independence of noise in individual
variables. Evidently, more general and versatile structures for the
noise statistics can be treated in a similar manner, and these may
become important when dealing with large databases.

A very important topic for future research is that of dealing with
statistical errors in estimating empirical statistics. It is common
to quantify distances using standard matrix norms --as is done in
the present paper as well. Alternative distance measures such as the
Wasserstein distance mentioned in Section~\ref{statisticalerrors}
and others (see e.g.,  \cite{ning2011}) may become increasingly
important in quantifying statistical uncertainty.

Finally, we raise the question of the asymptotic performance of certificates such as those presented in Section \ref{sec:CertifMinRank}. It is important to know how the tightness of the certificate to the minimal rank of linear models relates to the size of the problem.

\section*{Acknowledgments}

This work was supported in part by grants from NSF, NIH, AFOSR, ONR,
and MDA. This work is part of the National Alliance for Medical
Image Computing (NA-MIC), funded by the National Institutes of
Health through the NIH Roadmap for Medical Research, Grant U54
EB005149. Information on the National Centers for Biomedical
Computing can be obtained from
http://nihroadmap.nih.gov /bioinformatics. Finally, this project was
supported by grants from the National Center for Research Resources
(P41-RR-013218) and the  National Institute of Biomedical Imaging
and Bioengineering (P41-EB-015902) of the National Institutes of
Health.

\bibliographystyle{siam}
\bibliography{IEEEabrv,Frisch_bib}

\begin{thebibliography}{10}

\bibitem{Anderson1988}
{\sc B.~D.~O. Anderson and M.~Deistler}, {\em Identification of dynamic systems
  from noisy data}, Institute for Econometrics and Operations Research,
  Technical University, Vienna,  (1988).

\bibitem{Anderson2008}
\leavevmode\vrule height 2pt depth -1.6pt width 23pt, {\em Generalized linear
  dynamic factor models-a structure theory}, in 47th IEEE Conference on
  Decision and Control, 2008, pp.~1980--1985.

\bibitem{AndersonRubin}
{\sc T.~Anderson and H.~Rubin}, {\em Statistical inference in factor analysis},
  in Proceedings of the third Berkeley symposium on mathematical statistics and
  probability, vol.~5, 1956, pp.~111--150.

\bibitem{Bekker1997}
{\sc P.~A. Bekker and J.~M.~F. ten Berge}, {\em Generic global indentification
  in factor analysis}, Linear Algebra and its Applications, 264 (1997),
  pp.~255--263.

\bibitem{Boyd2004convex}
{\sc S.~Boyd and L.~Vandenberghe}, {\em Convex optimization}, Cambridge
  university press, 2004.

\bibitem{Candes2009exact}
{\sc E.~J. Cand{\`e}s and B.~Recht}, {\em Exact matrix completion via convex
  optimization}, Foundations of Computational Mathematics, 9 (2009),
  pp.~717--772.

\bibitem{Chandrasekaran2011rank}
{\sc V.~Chandrasekaran, S.~Sanghavi, P.~A. Parrilo, and A.~S. Willsky}, {\em
  Rank-sparsity incoherence for matrix decomposition}, SIAM Journal on
  Optimization, 21 (2011), pp.~572--596.

\bibitem{Deistler1989}
{\sc M.~Deistler and B.~D.~O. Anderson}, {\em Linear dynamic
  errors-in-variables models: Some structure theory}, Journal of Econometrics,
  41 (1989), pp.~39--63.

\bibitem{Della1982}
{\sc G.~Della~Riccia and A.~Shapiro}, {\em Minimum rank and minimum trace of
  covariance matrices}, Psychometrika, 47 (1982), pp.~443--448.

\bibitem{douglas}
{\sc R.~G. Douglas}, {\em On majorization, factorization, and range inclusion
  of operators on hilbert space}, Proceedings of the American Mathematical
  Society, 17 (1966), pp.~413--415.

\bibitem{Durbin}
{\sc J.~Durbin}, {\em Errors in variables}, Revue de l'Institut international
  de statistique, 22 (1954), pp.~23--32.

\bibitem{Eldar2004competitive}
{\sc Y.~Eldar and N.~Merhav}, {\em A competitive minimax approach to robust
  estimation of random parameters}, IEEE Transactions on Signal Processing, 52
  (2004), pp.~1931--1946.

\bibitem{Fazel2001}
{\sc M.~Fazel, H.~Hindi, and S.~P. Boyd}, {\em A rank minimization heuristic
  with application to minimum order system approximation}, in Proceedings of
  the 2001 American Control Conference, vol.~6, 2001, pp.~4734--4739.

\bibitem{Fazel2003}
\leavevmode\vrule height 2pt depth -1.6pt width 23pt, {\em Log-det heuristic
  for matrix rank minimization with applications to hankel and euclidean
  distance matrices}, in Proceedings of the 2003 American Control Conference,
  vol.~3, 2003, pp.~2156--2162.

\bibitem{Forni2000}
{\sc M.~Forni, M.~Hallin, M.~Lippi, and L.~Reichlin}, {\em The generalized
  dynamic-factor model: Identification and estimation}, Review of Economics and
  Statistics, 82 (2000), pp.~540--554.

\bibitem{Frisch2}
{\sc R.~Frisch}, {\em Statistical confluence analysis by means of complete
  regression systems}, vol.~5, Universitetets {\O}konomiske Instituut, 1934.

\bibitem{Guidorzi95}
{\sc R.~P. Guidorzi}, {\em Identification of the maximal number of linear
  relations from noisy data}, Systems \& control letters, 24 (1995),
  pp.~159--165.

\bibitem{Guttman1954}
{\sc L.~Guttman}, {\em Some necessary conditions for common-factor analysis},
  Psychometrika, 19 (1954), pp.~149--161.

\bibitem{Harman1966}
{\sc H.~Harman and W.~Jones}, {\em Factor analysis by minimizing residuals
  (minres)}, Psychometrika, 31 (1966), pp.~351--368.

\bibitem{Horn1990matrix}
{\sc R.~A. Horn and C.~R. Johnson}, {\em Matrix Analysis}, Cambridge University
  Press, 1990.

\bibitem{Joreskog1969}
{\sc K.~J{\"o}reskog}, {\em A general approach to confirmatory maximum
  likelihood factor analysis}, Psychometrika, 34 (1969), pp.~183--202.

\bibitem{Kalman1982}
{\sc R.~E. Kalman}, {\em System identification from noisy data}, in Dynamical
  Systems II, A.~Bednarek and L.~Cesari, eds., Academic Press, New York, 1982,
  pp.~135--164.

\bibitem{Kalman1985}
\leavevmode\vrule height 2pt depth -1.6pt width 23pt, {\em Identification of
  noisy systems}, Russian Mathematical Surveys, 40 (1985), p.~25.

\bibitem{Keshavan2010matrix}
{\sc R.~H. Keshavan, A.~Montanari, and S.~Oh}, {\em Matrix completion from a
  few entries}, IEEE Transactions on Information Theory, 56 (2010),
  pp.~2980--2998.

\bibitem{Keshavan2010noisy}
\leavevmode\vrule height 2pt depth -1.6pt width 23pt, {\em Matrix completion
  from noisy entries}, The Journal of Machine Learning Research, 99 (2010),
  pp.~2057--2078.

\bibitem{KlepperLeamer}
{\sc S.~Klepper and E.~E. Leamer}, {\em Consistent sets of estimates for
  regressions with errors in all variables}, Econometrica: Journal of the
  Econometric Society, 52 (1984), pp.~163--183.

\bibitem{Koopmans}
{\sc T.~C. Koopmans}, {\em Linear regression analysis of economic time series},
  Netherlands Economic Institute, Harrlem-de Erwen F. Bohn N.V., 1937.

\bibitem{Ledermann1939}
{\sc L.~L. Ledermann}, {\em On a problem concerning matrices with variable
  diagonal elements}, Proceedings of the Royal Society of Edinburgh, 60 (1940),
  pp.~1--17.

\bibitem{Ledermann}
{\sc W.~Ledermann}, {\em On the rank of the reduced correlational matrix in
  multiple-factor analysis}, Psychometrika, 2 (1937), pp.~85--93.

\bibitem{Ledermann1940}
\leavevmode\vrule height 2pt depth -1.6pt width 23pt, {\em On a problem
  concerning matrices with variable diagonal elements}, Williams and Norgate,
  1940.

\bibitem{Los}
{\sc C.~A. Los}, {\em Identification of a linear system from inexact data: a
  three-variable example}, Computers \& Mathematics with Applications, 17
  (1989), pp.~1285--1304.

\bibitem{ning2011}
{\sc L.~Ning, X.~Jiang, and T.~Georgiou}, {\em Geometric methods for estimation
  of structured covariances}, arXiv:1110.3695,  (2011).

\bibitem{olkin1982}
{\sc I.~Olkin and F.~Pukelsheim}, {\em The distance between two random vectors
  with given dispersion matrices}, Linear Algebra and Its Applications, 48
  (1982), pp.~257--263.

\bibitem{Recht2010guaranteed}
{\sc B.~Recht, M.~Fazel, and P.~A. Parrilo}, {\em Guaranteed minimum-rank
  solutions of linear matrix equations via nuclear norm minimization}, SIAM
  review, 52 (2010), pp.~471--501.

\bibitem{Reiersol}
{\sc O.~Reiers{\o}l}, {\em Confluence analysis by means of lag moments and
  other methods of confluence analysis}, Econometrica: Journal of the
  Econometric Society, 9 (1941), pp.~1--24.

\bibitem{Saunderson2012}
{\sc J.~Saunderson, V.~Chandrasekaran, P.~Parrilo, and A.~Willsky}, {\em
  Diagonal and low-rank matrix decompositions, correlation matrices, and
  ellipsoid fitting}, arXiv:1204.1220,  (2012).

\bibitem{Shapiro}
{\sc A.~Shapiro}, {\em Rank-reducibility of a symmetric matrix and sampling
  theory of minimum trace factor analysis}, Psychometrika, 47 (1982),
  pp.~187--199.

\bibitem{Shapiro1982b}
\leavevmode\vrule height 2pt depth -1.6pt width 23pt, {\em Weighted minimum
  trace factor analysis}, Psychometrika, 47 (1982), pp.~243--264.

\bibitem{Shapiro1985}
\leavevmode\vrule height 2pt depth -1.6pt width 23pt, {\em Identifiability of
  factor analysis: Some results and open problems}, Linear algebra and its
  applications, 70 (1985), pp.~1--7.

\bibitem{Soderstrom2007errors}
{\sc T.~S{\"o}derstr{\"o}m}, {\em Errors-in-variables methods in system
  identification}, Automatica, 43 (2007), pp.~939--958.

\bibitem{Spearman}
{\sc C.~Spearman}, {\em General intelligence, objectively determined and
  measured}, The American Journal of Psychology, 15 (1904), pp.~201--292.

\bibitem{vantrees}
{\sc H.~L. Van~Trees}, {\em Optimum array processing}, Wiley-Interscience,
  2002.

\bibitem{Varga2004}
{\sc R.~Varga}, {\em Ger{\v{s}}gorin and his circles}, Springer Verlag, 2004.

\bibitem{Woodgate1}
{\sc K.~G. Woodgate}, {\em An upper bound on the number of linear relations
  identified from noisy data by the {F}risch scheme}, Systems \& control
  letters, 24 (1995), pp.~153--158.

\bibitem{woodgate2}
\leavevmode\vrule height 2pt depth -1.6pt width 23pt, {\em On computing the
  maximum corank in the {F}risch scheme}, Citeseer; Pre-print 4 pages,  (2007).

\end{thebibliography}
\end{document}